\documentclass[12pt,technote, onecolumn, draft]{IEEEtran}

\usepackage{amsmath,amsfonts}
\usepackage{algorithmic}
\usepackage{array}
\usepackage[caption=false,font=normalsize,labelfont=sf,textfont=sf]{subfig}
\usepackage{textcomp}
\usepackage{stfloats}
\usepackage{url}
\usepackage{verbatim}
\usepackage{graphicx}
\usepackage{amsmath}    
\usepackage{amsthm} 
\usepackage{mathrsfs}
\usepackage{amssymb}
\usepackage{cite}
\usepackage{algorithm}        
\usepackage{algorithmic}
\usepackage{amsmath,amssymb}
\usepackage{tikz}
\usetikzlibrary{positioning, arrows.meta}
\usepackage{caption} 
\usepackage{hyperref}
\usepackage{cleveref}
 
\newtheorem{theorem}{Theorem}[section]
\newtheorem{corollary}[theorem]{Corollary}
\newtheorem{definition}{Definition}
\newtheorem{proposition}{Proposition}
\newtheorem{lemma}{Lemma}
\newtheorem{remark}{Remark}
\newtheorem{example}{Example}

\ifCLASSINFOpdf
\else
\fi

\begin{document}

\title{Secure Network Function Computation for General Target and Security  Functions$^{\dag}$}

\author{Qin Zhou, Fang-Wei Fu
\IEEEcompsocitemizethanks{\IEEEcompsocthanksitem Qin Zhou and Fang-Wei Fu are with Chern Institute of Mathematics and LPMC, Nankai University, Tianjin 300071, China, Emails: qinzhou@mail.nankai.edu.cn, fwfu@nankai.edu.cn.
}
\thanks{$^\dag$This research is supported by the National Key Research and Development Program of China (Grant No. 2022YFA1005000), the National Natural Science Foundation of China (Grant No.  62371259), the Fundamental Research Funds for the Central Universities of China (Nankai University), and the Nankai Zhide Foundation.}
}
\maketitle

\begin{abstract}
Secure network function computation is a critical research direction in network coding, which aims to ensure that the target function is correctly computed at the sink node while preventing the wiretapper from obtaining any information about the security function. In this paper, we focus on the general secure network function computation model, where the target function $f$ and the security function $\zeta$ are arbitrary, and the wiretapper can eavesdrop on any subset of edges with size at most a given security level. Using information-theoretic techniques, we establish a nontrivial upper bound on the secure computing capacity, which is applicable to arbitrary networks, arbitrary target and security functions, and arbitrary security levels. This upper bound is shown to degenerate to the existing bounds in the literature when the target and security functions are specific forms. Furthermore, we consider two specific models: one where the target function is vector-linear and the security function is the identity function, and another where both functions are vector-linear. For the former, we derive a simplified form of the upper bound on the secure computing capacity via order-theoretic methods and propose an efficient algorithm to compute this bound with linear time complexity in the number of network edges. For the latter, we characterize the equivalent conditions for the computability and security of linear secure network codes, develop two constructive schemes for such codes, and derive an upper bound on the minimal finite field size required for the constructions, thereby obtaining a nontrivial lower bound on the secure computing capacity. 
\end{abstract}


\section{Introduction}
Network function computation arises as a natural extension of network coding, a groundbreaking paradigm introduced by Ahlswede et al. \cite{ANLY2000} that revolutionized data transmission by allowing intermediate nodes to linearly combine incoming data flows. Unlike traditional network coding, which focuses on the error-free recovery of source messages, network function computation aims to enable sink nodes to accurately compute predefined target functions of messages from multiple sources over general networks \cite{ARF2011,ARF2013,ARF2014,HTY2018,GYY2019,LX2022,ZF2024,ZGY2024,GZ2024}. This generalization has found widespread applications in important fields such as sensor network data fusion \cite{GK2005}, distributed storage systems, large-scale industrial data processing \cite{DP2016}, and machine learning \cite{VJ2016}.

As distributed computing and large-scale data transmission become increasingly pervasive, security has become an indispensable requirement for network function computation systems. To address this concern, information-theoretic security has been widely studied as a critical security measure, as it guarantees information-theoretic privacy independent of the eavesdropper’s computational power. Its theoretical foundations can be traced to several pioneering works, including the Shannon cipher system \cite{S1949}, secret sharing \cite{BG1979,S1979}, and Wiretap Channel II \cite{OW1985}. These classic frameworks have laid a solid theoretical basis for secure data transmission and computation in networked systems. Cai and Yeung \cite{CY2011} first studied the information-theoretic security problem in networks with eavesdroppers, which is known as secure network coding. Its core requirement is that the sink node can reliably recover source messages, while the wiretapper obtains no meaningful information about them. Building on this line of work, Guang et al. \cite{GBY2024} proposed the problem of secure network function computation, whose core requirement is to ensure that the sink node computes the target function with zero error, while preventing the wiretapper from obtaining any information about the protected security function even after eavesdropping on any edge set in the wiretap collection. This work has established a fundamental framework for subsequent research.

\subsection{Related Works}

When the wiretap collection is empty, secure network function computation reduces to classical network function computation, which has been extensively studied \cite{ARF2011,ARF2013,ARF2014,HTY2018,GYY2019,LX2022,ZF2024,ZGY2024,GZ2024}. In this model, a single sink computes a target function of source messages generated at multiple source nodes over a directed acyclic network. When the target function is the identity function, the problem further reduces to the classical network coding problem \cite{Y2008,YLC2005,FS2007,2FS2007,HL2008}. Appuswamy et al. \cite{ARF2011} investigated the computing capacity, defined as the maximum average number of times the target function can be computed error-free per network use, and derived a cut-set upper bound constrained by network topology and target function properties. Huang et al. \cite{HTY2018} proposed a new cut-set upper bound applicable to arbitrary target functions and network topologies, and showed that this bound is tight for computing arbitrary functions over multi-edge tree networks and for computing the identity or algebraic sum functions over arbitrary networks. Furthermore, Guang et al. \cite{GYY2019} established an improved upper bound on computing capacity using cut-set strong partitioning, which also applies to arbitrary functions and network topologies. Considerable attention has also been devoted to network function computation with vector linear target functions \cite{ARF2011,ARF2013,LX2022,ZF2024}. However, exactly characterizing the computing capacity remains extremely difficult for general functions and networks. Even for linear target functions, only a small class of special cases has been solved, while most others remain open \cite{ARF2014}. Another important direction considers link errors, leading to the robust network function computation problem. Wei et al. \cite{WXG2023} studied the robust computing capacity, i.e., the maximum average number of times the target function can be computed with zero error per network use when at most $\tau$ links are corrupted by errors, and derived corresponding cut-set upper bounds. They also defined the minimum distance of linear network codes for linear target functions to characterize error-correction capability, and established a Singleton-type bound for this minimum distance.

The study of secure network function computation is rooted in both information-theoretic security and the framework of network function computation. A general model is formulated as follows: in a directed acyclic single-sink network $\mathcal{N}$, the sink node is required to repeatedly compute a target function $f$ with zero error from messages generated at multiple sources. Meanwhile, an eavesdropper can wiretap any edge set in a given collection $\mathcal{W}$, but must not learn any information about a security function $\zeta$ of the source messages. When both $f$ and $\zeta$ are identity functions, the model reduces to classical secure network coding \cite{CY2011,ESS2012,GY2018,SK2011,GYF2020}. Guang et al. \cite{GBY2024} formalized secure network function computation and introduced the secure computing capacity, defined as the maximum average number of times the target function $f$ can be securely computed per network use under given wiretap sets and security functions. Their work focused on scalar linear target functions, wiretap sets consisting of all edge subsets of size at most $r$, and the identity security function, and derived nontrivial upper and lower bounds on the secure computing capacity. Subsequently, Bai et al. \cite{arxiv2025secure} extended the setup to scalar linear target functions and security functions identical to the target function, also deriving nontrivial bounds. In this paper, we study secure network function computation without restricting to linear functions.

\subsection{Contributions and Organization of the Paper}
This paper builds upon the secure network function computation model $(\mathcal{N},f,\mathcal{W},\zeta)$ proposed in \cite{GBY2024} and \cite{arxiv2025secure}, where the wiretap collection $\mathcal{W}$ consists of all edge subsets of size at most $r$, referred to as the security level.
From an information-theoretic viewpoint, our main goal is to characterize the secure computing capacity, defined as the maximum average number of times the target function $f$ can be computed with zero error for one use of the network, while ensuring that the wiretapper learns nothing about the security function from any single wiretap set.
Characterizing this capacity is highly challenging: even for classical network function computation without security constraints, the exact capacity remains open in general \cite{GYY2019}. Therefore, this paper focuses on establishing nontrivial upper and lower bounds for the secure computing capacity.

Our key contributions are summarized as follows:
First, we derive a nontrivial upper bound on the secure computing capacity that applies to arbitrary network topologies, arbitrary target and security functions, and arbitrary security levels. This bound generalizes those given in \cite{GBY2024} and \cite{arxiv2025secure}. Since the upper bound does not have a closed-form expression, we further develop closed-form upper and lower bounds for it. In particular, we identify several special cases where the secure computing capacity can be exactly determined, and obtain a nontrivial upper bound on the security level that allows positive-rate secure function computation.

Second, we consider vector-linear target functions and the identity security function. The upper bound on the secure computing capacity depends on network topology, security level, and the target function. For large-scale networks, high security levels, and target functions with dense coefficient matrices, its computational complexity becomes prohibitive. To overcome this issue, we use lattice-theoretic methods to compare the contributions of different edge subsets, yielding an equivalent form of the bound with significantly reduced complexity. Based on existing algorithms, we further develop an efficient algorithm to compute this bound with time complexity linear in the number of edges.

Finally, we study linear secure network coding when both the target and security functions are vector-linear. For given linear functions over $\mathbb{F}_{q}$, we propose two general constructions of linear secure network codes that compute the target function while protecting the security function. Using these constructions, we establish a lower bound on the secure computing capacity. Motivated by practical applications, we also analyze the required field size for the constructed codes and provide a sufficient condition for their existence. An illustrative example is included to demonstrate our code designs.

The remainder of the paper is organized as follows. Section~\ref{sec:preliminaries} introduces notation and definitions, formalizes the secure network function computation model and function-computing secure network codes, and reviews relevant lattice theory. Section~\ref{sec:upper_bound} establishes the general upper bound on the secure computing capacity. Section~\ref{sec:algorithm} presents the efficient algorithm for computing the upper bound. Section~\ref{sec:linear_codes} focuses on linear secure network codes, proposes two construction frameworks, and derives the lower bound on the secure computing capacity and an upper bound on the minimum required field size for our construction. Section~\ref{sec:conclusion} concludes the paper and discusses future research directions.

\section{Preliminaries}\label{sec:preliminaries}

\subsection{Secure Network Function Computation Model}

Let $\mathcal{G}=(\mathcal{V}, \mathcal{E})$ denote a directed acyclic graph, where $\mathcal{V}$ is a finite set of nodes and $\mathcal{E}$ is a finite set of edges. Multiple edges between two nodes are permitted. For each edge $e \in \mathcal{E}$, $\operatorname{tail}(e)$ and $\operatorname{head}(e)$ denote its tail node and head node, respectively. For any node $u \in \mathcal{V}$, the sets of input edges and output edges of $u$ are defined as $\operatorname{In}(u)=\{e \in \mathcal{E}: \operatorname{head}(e)=u\}$ and $\operatorname{Out}(u)=\{e \in \mathcal{E}: \operatorname{tail}(e)=u\}$, respectively.

A sequence of edges $(e_{1}, e_{2}, \cdots, e_{m})$ in $\mathcal{G}$ is referred to as a path from node $u$ (or edge $e_{1}$) to node $v$ (or edge $e_{m}$) if it satisfies $\operatorname{tail}(e_{1})=u$, $\operatorname{head}(e_{m})=v$, and $\operatorname{tail}(e_{i})=\operatorname{head}(e_{i-1})$ for $i=2,3, \cdots, m$. In particular, a single edge $e$ is a path from $\operatorname{tail}(e)$ to $\operatorname{head}(e)$. We write $d \to e$ (or $\operatorname{tail}(d) \to \operatorname{head}(e)$, $\operatorname{tail}(d) \to e$, $d \to \operatorname{head}(e)$) if there exists a path from edge $d$ to edge $e$.

For two disjoint subsets $U$ and $V$, where $U$ ($V$) can be either a node set or an edge set, an edge subset $C \subseteq \mathcal{E}$ is called a cut separating $V$ from $U$ if, for each pair $(u,v) \in U \times V$, no path exists from $u$ to $v$ after removing all edges in $C$. In particular, if $U=\{u\}$ and $V=\{v\}$ are two disjoint singleton subsets, a cut separating $V$ from $U$ is also called a cut separating $v$ from $u$. The capacity of a cut is defined as its size, i.e., the number of edges it contains. A cut $C$ separating $V$ from $U$ is called a minimum cut separating $V$ from $U$ if there does not exist a cut $C'$ separating $V$ from $U$ such that $|C'| < |C|$.
Its capacity is called the minimum cut capacity separating $V$ from $U$, denoted by $\operatorname{mincut}(U, V)$. The set of all minimum cuts separating $V$ from $U$ is denoted as $\operatorname{MINCUT}(U, V )$. For $U=\{u\}$ and $V=\{v\}$, this is simplified to $\operatorname{mincut}(u, v)$ and $\operatorname{MINCUT}(u, v)$. 

Let $S \subset \mathcal{V}$ be the set of source nodes $\sigma_{1}, \sigma_{2}, \cdots, \sigma_{s}$, and let $\rho \in \mathcal{V}\setminus S$ be the single sink node, where $\operatorname{In}(\sigma_{i})=\emptyset$ (i.e., source nodes have no input edges) and $\operatorname{Out}(\rho)=\emptyset$ (i.e., the sink node has no output edges) for all $i=1,2, \cdots, s$. We assume there exists a directed path from every node $u \in \mathcal{V}$ to $\rho$. 
The graph $\mathcal{G}$, together with the source set $S$ and the sink $\rho$, forms a network $\mathcal{N}$, denoted by $\mathcal{N}=(\mathcal{G}, S, \rho)$.

Let $\mathcal{A}_{i}$ ($1 \leq i \leq s$), $\mathcal{O}$, $\mathcal{B}$, and $\mathcal{Q}$ be finite alphabets. We assume that a symbol from $\mathcal{B}$ can be reliably transmitted on each edge per use, i.e., we take the capacity of each edge to be $1$ with respect to the alphabet $\mathcal{B}$. Let $M_{i}$ denote the random variable representing the information source at the $i$-th source node $\sigma_{i}$, which follows a uniform distribution over $\mathcal{A}_{i}$. All $M_{i}$ ($1 \leq i \leq s$) are mutually independent. Let $M_{S}=(M_{1}, M_{2}, \cdots, M_{s})$. For positive integers $\ell$ and $n$, let $M_{i,j}$ ($1 \leq j \leq \ell$) be independent and identically distributed (i.i.d.) random variables with the generic random variable $M_{i}$, which is sequentially generated by the $i$-th source node $\sigma_{i}$. Let $\mathbf{M}_{i}=(M_{i,1}, M_{i,2}, \cdots, M_{i,\ell})$ be the source message generated by $\sigma_{i}$. Further let $\mathbf{M}_{S}=(\mathbf{M}_{1}, \mathbf{M}_{2}, \cdots, \mathbf{M}_{s})$ be the source message vector generated by $S$.

Let $f: \prod_{i=1}^{s} \mathcal{A}_{i} \to \mathcal{O}$ be a nonconstant target function, which is required to be computed with zero error at the sink node $\rho$. Without loss of generality, we assume the $i$-th argument of $f$ is generated at the $i$-th source node $\sigma_{i}$. When computing $f$ $\ell$ times using the network $n$ times (i.e., transmitting at most $n$ symbols from $\mathcal{B}$ on each edge), the $\ell$ values of $f$
\[f\left(\mathbf{M}_{S}\right) \triangleq\left(f\left(M_{1,j}, M_{2,j}, \cdots, M_{s,j}\right): j=1,2, \cdots, \ell\right)\]
must be computed at $\rho$ with zero error.

Let $\mathcal{W}$ be a collection of edge subsets, where each $W \in \mathcal{W}$ is called a wiretap set. Let $\zeta: \prod_{i=1}^{s} \mathcal{A}_{i} \to \mathcal{Q}$ be a nonconstant security function that specifies the information to be protected from the eavesdropper. When $f(\mathbf{M}_{S})$ is computed at $\rho$, the $\ell$ values of $\zeta$
\[\zeta(\mathbf{M}_{S})\triangleq \left( \zeta (M_{1,j},M_{2,j},\cdots ,M_{s,j}): j=1,2,\cdots ,\ell \right)\]
must be protected from an eavesdropper who can access any one, but not more than one, wiretap set $W \in \mathcal{W}$. The collection of wiretap sets $\mathcal{W}$ and the security function $\zeta$ are known to all source and sink nodes, while the specific wiretap set accessed by the eavesdropper remains unknown. The formal model of secure network function computation is denoted by $(\mathcal{N}, f, \mathcal{W}, \zeta)$, where $\mathcal{N}$ is the network, $f$ is the target function, $\mathcal{W}$ is the collection of wiretap sets, and $\zeta$ is the security function. 

\subsection{Function-Computing Secure Network Code}

In this subsection, we present function-computing secure network codes for the secure network function computation model $(\mathcal{N}, f, \mathcal{W},\zeta)$. Randomizing source messages is a highly effective technique for ensuring information-theoretic security. To
this end, let $\mathbf{K}_{i}$ ($1 \leq i \leq s$) be a random variable (referred to as a key) available to the $i$-th source node $\sigma_{i}$, which follows a uniform distribution over a finite alphabet $\mathcal{K}_{i}$. Let $\mathbf{K}_{S}=(\mathbf{K}_{1}, \mathbf{K}_{2}, \cdots, \mathbf{K}_{s})$ denote the key vector. Assume that all keys $\mathbf{K}_{i}$ and source messages $\mathbf{M}_{i}$ are mutually independent.

An $(\ell, n)$ function-computing secure network code $\hat{C}$ for $(\mathcal{N}, f, \mathcal{W},\zeta)$ is defined as follows. Let $\mathbf{m}_{i} \in \mathcal{A}_{i}^{\ell}$ and $\mathbf{k}_{i}\in \mathcal{K}_{i}$ be arbitrary realizations of the source message $\mathbf{M}_{i}$ and the key $\mathbf{K}_{i}$, respectively, for $i=1,2,\cdots,s$. Let $\mathbf{m}_{S}=(\mathbf{m}_{1}, \mathbf{m}_{2}, \cdots, \mathbf{m}_{s})$ and $\mathbf{k}_{S}=(\mathbf{k}_{1}, \mathbf{k}_{2}, \cdots, \mathbf{k}_{s})$ be the corresponding realizations of $\mathbf{M}_{S}$ and $\mathbf{K}_{S}$. The code $\hat{C}$ consists of a local encoding function $\hat{\theta}_{e}$ for each edge $e\in \mathcal{E}$, where
\begin{align}\label{eq:local encoding function}
\hat{\theta}_{e}:\begin{cases} 
\mathcal{A}_{i}^{\ell} \times \mathcal{K}_{i} \to \mathcal{B}^{n}, & \text{if } \operatorname{tail}(e)=\sigma_{i} \text{ for some } i, \\
\prod_{d \in \operatorname{In}(\operatorname{tail}(e))} \mathcal{B}^{n} \to \mathcal{B}^{n}, & \text{otherwise};
\end{cases}
\end{align}
and a decoding function $\hat{\varphi}: \prod_{\operatorname{In}(\rho)} \mathcal{B}^{n} \to \mathcal{O}^{\ell}$ at the sink node $\rho$, which is used to compute $f$ with zero error. Let $\mathbf{y}_{e} \in \mathcal{B}^{n}$ be the message transmitted on edge $e$ using the code $\hat{C}$ under the source message $\mathbf{m}_{S}$ and the key $\mathbf{k}_{S}$. With the
encoding mechanism described in \eqref{eq:local encoding function}, $\mathbf{y}_{e}$ is a function of $\mathbf{m}_{S}$ and $\mathbf{k}_{S}$, denoted by $\hat{g}_{e}(\mathbf{m}_{S},\mathbf{k}_{S})$, which is obtained by recursively applying the local encoding functions $\hat{\theta}_{e}$. Formally,
\[\hat{g}_{e}\left(\mathbf{m}_{S}, \mathbf{k}_{S}\right)=
\begin{cases} 
\hat{\theta}_{e}\left(\mathbf{m}_{i}, \mathbf{k}_{i}\right), & \text{if } \operatorname{tail}(e)=\sigma_{i} \text{ for some } i, \\
\hat{\theta}_{e}\left(\hat{g}_{\operatorname{In}(u)}\left(\mathbf{m}_{S}, \mathbf{k}_{S}\right)\right), & \text{otherwise},
\end{cases}\]
where $u=\operatorname{tail}(e)$, and $\hat{g}_{E}(\mathbf{m}_{S}, \mathbf{k}_{S})=(\hat{g}_{e}(\mathbf{m}_{S}, \mathbf{k}_{S}): e \in E)$ for any edge subset $E \subseteq \mathcal{E}$. The function $\hat{g}_{e}$ is called the global encoding function of edge $e$ for the code $\hat{C}$.

An $(\ell, n)$ secure network code $\hat{C}=\{\hat{\theta}_{e}: e \in \mathcal{E}; \hat{\varphi}\}$ is admissible if it satisfies the following two conditions:
\begin{itemize}
\item \textbf{Computability Condition:} The sink node $\rho$ computes $f$ with zero error, i.e., for all $\mathbf{m}_{S} \in \prod_{i=1}^{s} \mathcal{A}_{i}^{\ell}$ and $\mathbf{k}_{S} \in \prod_{i=1}^{s} \mathcal{K}_{i}$,
\begin{align}\label{eq:computability}
\hat{\varphi}\left(\hat{g}_{\operatorname{In}(\rho)}\left(\mathbf{m}_{S}, \mathbf{k}_{S}\right)\right)=f\left(\mathbf{m}_{S}\right);
\end{align}
\item \textbf{Security Condition:} For any wiretap set $W \in \mathcal{W}$, the transmitted messages on $W$ are independent of $\zeta(\mathbf{M}_{S})$, i.e.,
\begin{align}\label{eq:security condition}
I(Y_{W}; \zeta(\mathbf{M}_{S}))=0,
\end{align}
where $Y_{W}=(Y_{e}: e \in W)$ and $Y_{e} \triangleq \hat{g}_{e}(\mathbf{M}_{S}, \mathbf{K}_{S})$ is the random vector transmitted on edge $e$.
\end{itemize}

The secure computing rate of an admissible $(\ell, n)$ code $\hat{C}$ is defined as
\[R(\hat{C}) \triangleq \frac{\ell}{n},\] 
which represents the average number of times $f$ can be computed with zero error for one use of the network $\mathcal{N}$, while no information about the security function $\zeta$ is leaked to the wiretapper who can access at most one wiretap set $W \in \mathcal{W}$. A nonnegative real number $R$ is achievable if, for any $\epsilon>0$, there exists an admissible $(\ell, n)$ code $\hat{C}$ such that 
\[R(\hat{C})>R-\epsilon.\]
The secure computing rate region for the model $(\mathcal{N},f,\mathcal{W},\zeta)$ is the set of all achievable rates, i.e.,
\[
\mathfrak{R}(\mathcal{N}, f, \mathcal{W},\zeta)=\{R: R \text{ is achievable for } (\mathcal{N},f,\mathcal{W},\zeta)\}.
\]
The secure computing capacity for the model $(\mathcal{N},f,\mathcal{W},\zeta)$ is the maximum value of the rate region, i.e.,
\[\hat{\mathcal{C}}(\mathcal{N}, f, \mathcal{W},\zeta) \triangleq \max \mathfrak{R}(\mathcal{N}, f, \mathcal{W},\zeta).\]

In this paper, we consider a special collection of wiretap sets
\[ \mathcal{W}_r \triangleq \{ W \subseteq \mathcal{E} : 0 \leq |W| \leq r \}, \]
which means that the wiretapper can eavesdrop on the information transmitted over any subset of edges in the network with size no more than $r$. Here, $r$ is called the security level. Notably, the empty set $\emptyset$, which corresponds to the scenario where no edges are eavesdropped, is included in $\mathcal{W}_r$ as a wiretap set of size $0$.

\subsection{Edge-Subset Lattice}

This subsection introduces some basic concepts and properties of lattices in order theory (cf. \cite{Davey1990}, \cite{Birkhoff1967}), which provides a useful tool for analyzing the structure of edge subsets and their interactions in networks.

\begin{definition}\label{def:poset}
Let $P$ be a nonempty set. A non-strict partial order ``$\leq$" on $P$ is a binary relation that satisfies the following three conditions for all $a, b, c \in P$:
\begin{itemize}
\item (Reflexivity) $a \leq a$;
\item (Antisymmetry) If $a \leq b$ and $b \leq a$, then $a=b$;
\item (Transitivity) If $a \leq b$ and $b \leq c$, then $a \leq c$.
\end{itemize}
\end{definition}

The set $P$ equipped with the partial order ``$\leq$" is called a partially ordered set, denoted by $(P, \leq)$. For a partially ordered set $(P, \leq)$ and a subset $Q \subseteq P$, an element $a \in P$ is called a lower bound of $Q$ if $a \leq x$ for all $x \in Q$. Similarly, an element $a \in P$ is called an upper bound of $Q$ if $x \leq a$ for all $x \in Q$. The sets of all lower and upper bounds of $Q$ are denoted by $Q^{l}$ and $Q^{u}$, respectively, i.e.,
\[
Q^{l}=\{a \in P: a \leq x, \forall x \in Q\},
\]
and
\[
Q^{u}=\{a \in P: x \leq a, \forall x \in Q\}.
\]
A lower bound $a^{*} \in Q^{l}$ is called the infimum of $Q$ if $a \leq a^{*}$ for all $a \in Q^{l}$. Dually, an upper bound $a^{*} \in Q^{u}$ is the supremum of $Q$ if $a^{*} \leq a$ for all $a \in Q^{u}$. Note that the infimum and supremum of $Q$ may not exist. From an algebraic perspective, the infimum of $Q$ is also called the meet, denoted by $\wedge Q$; the supremum is also called the join, denoted by $\vee Q$. For a two-element subset $\{a, b\} \subseteq P$, the meet and join are simplified to $a \wedge b$ and $a \vee b$, respectively.

\begin{definition}\label{def:semilattice-lattice}
A partially ordered set $(P, \leq)$ is a meet-semilattice if the meet $a \wedge b$ exists for any two elements $a, b \in P$. Similarly, it is a join-semilattice if the join $a \vee b$ exists for any two elements $a, b \in P$. A partially ordered set that is both a meet-semilattice and a join-semilattice is called a lattice.
\end{definition}

\begin{definition}\label{def:bottom-top}
For a partially ordered set $(P, \leq)$, an element $\perp \in P$ is called a bottom if $\perp \leq a$ for all $a \in P$. Similarly, an element $\top \in P$ is called a top if $a \leq \top$ for all $a \in P$.
\end{definition}

By definition, the bottom and top of a partially ordered set (if they exist) are unique. The following proposition characterizes the existence of the bottom and top in finite semilattices and lattices.

\begin{proposition}\label{prop:bottom-top-exist}
A finite meet-semilattice has a unique bottom $\perp$; dually, a finite join-semilattice has a unique top $\top$. Consequently, a finite lattice has both a unique bottom $\perp$ and a unique top $\top$.
\end{proposition}

\section{Upper Bound on the Secure Computing Capacity}\label{sec:upper_bound}

In this section, we consider the secure network function computation model $(\mathcal{N},f,\mathcal{W}_{r},\zeta)$ and present an upper bound on the secure computing capacity $\hat{\mathcal{C}}(\mathcal{N}, f, \mathcal{W}_{r},\zeta)$. To establish this result, we first introduce several key definitions and properties related to common functions and decomposable pairs, which serve as fundamental tools for deriving the main result.

First, we introduce several graph-theoretic notations. Given an edge subset $C \subseteq \mathcal{E}$, we define three subsets of source nodes as follows:
\begin{align*}
 D_{C}&=\{\sigma \in S: \exists \, e \in C \text{ such that } \sigma \rightarrow e\}, \\
 I_{C}&=\{\sigma \in S: \sigma \nrightarrow \rho \text{ after deleting the edges in } C \text{ from } \mathcal{E}\}, \\
 J_{C}&=D_{C} \backslash I_{C},
\end{align*}
where $\sigma \nrightarrow \rho$ indicates that there is no directed path from $\sigma$ to $\rho$. Recall that we assume there exists a directed path from every node $u \in \mathcal{V} \backslash \{\rho\}$ to $\rho$, and in particular, from every source node $\sigma \in S$ to $\rho$. This implies that $I_{C} \subseteq D_{C}$. Intuitively, $J_{C}$ denotes the subset of source nodes $\sigma$ for which there exist both a path from $\sigma$ to $\rho$ passing through at least one edge in $C$ and a path from $\sigma$ to $\rho$ that does not pass through any edge in $C$. In the network $\mathcal{N}$, an edge set $C$ is called a cut set if $I_{C} \neq \emptyset$, and we let $\Lambda(\mathcal{N})$ be the family of all cut sets, i.e.,
\[\Lambda(\mathcal{N})=\left\{C \subseteq \mathcal{E}: I_{C} \neq \emptyset\right\}.\]
In particular, a cut set $C$ is referred to as a global cut set if $I_{C}=S$.

\begin{definition}
Let $g_{1}$ and $g_{2}$ be two random functions defined on a common probability space $\Omega$.
A function \(h\) defined on $\Omega$ is called a common function of \(g_{1}\) and \(g_{2}\) if it satisfies:
\begin{itemize}
    \item $H(h\mid g_{1})=0$, i.e., $h$ is a deterministic function of $g_{1}$;
    \item $H(h\mid g_{2})=0$, i.e., $h$ is a deterministic function of $g_{2}$.
\end{itemize}
The set of all common functions of \(g_{1}\) and \(g_{2}\) is denoted by $g_{1}\cap g_{2}$. The maximal common function of \(g_{1}\) and \(g_{2}\), denoted by $g_{1}\sqcap g_{2}$, is the function in $g_{1}\cap g_{2}$ that attains the maximum entropy.
\end{definition}

\begin{definition}
Let \(h\) be a function of the random vector \(M_S\), and let \(C\in\Lambda(\mathcal{N})\) be a cut set.
A pair \((C,h)\) is called a \emph{decomposable pair} if \(h\) can be decomposed as
\begin{align}\label{eq:strong}
h(M_S)=\hat{h}\bigl(h_C(M_{I_C}),M_{S\setminus I_C}\bigr),
\end{align}
where \(h_C\) is a function only of \(M_{I_C}\).

Furthermore, if for every fixed \(M_{S\setminus I_C}\), the outer function \(\hat{h}(\,\cdot\,,M_{S\setminus I_C})\) is injective in its first argument, then \((C,h)\) is called a \emph{strongly decomposable pair}. In this case, Equation \eqref{eq:strong} is referred to as a strong decomposition of $(C,h)$, and the corresponding function \(h_C\) is called the intermediate function of this strong decomposition.
\end{definition}

\begin{lemma}\label{lem:strong decomposable}
For a fixed strongly decomposable pair \((C,h)\),
let \(g\) be a function of $M_S$.
Although the intermediate function \(h_C\) of the strong decomposition may not be unique,
the set of common functions \(h_C\cap g\) is invariant among all valid choices of \(h_C\).
\end{lemma}
\begin{proof}
Let \(h_C\) and \(h'_C\) be two intermediate functions satisfying the strong decomposition condition for \((C,h)\), i.e.,
\[
h(M_S) = \hat{h}\bigl(h_C(M_{I_C}), M_{S \setminus I_C}\bigr) = \hat{h}'\bigl(h'_C(M_{I_C}), M_{S \setminus I_C}\bigr).
\]
From this, we have
\begin{align}\label{eq:lemma1}
H(h(M_S) \mid h'_C(M_{I_C}), M_{S \setminus I_C}) = 0.
\end{align}
Since the outer function \(\hat{h}(\,\cdot\,, M_{S \setminus I_C})\) is injective in its first argument for a fixed \(M_{S \setminus I_C}\), it follows that
\begin{align}\label{eq:lemma2}
H(h_C(M_{I_C}) \mid h(M_S), M_{S \setminus I_C}) = 0.
\end{align}
Combining \eqref{eq:lemma1} and \eqref{eq:lemma2}, we obtain
\[
H(h_C(M_{I_C}) \mid h'_C(M_{I_C}), M_{S \setminus I_C}) = 0.
\]
Consequently, we have
\begin{align}
H(h_C(M_{I_C}) \mid h'_C(M_{I_C})) &= I(h_C(M_{I_C}) ; M_{S \setminus I_C} \mid h'_C(M_{I_C})) \notag\\
&= H(M_{S \setminus I_C} \mid h'_C(M_{I_C})) - H(M_{S \setminus I_C} \mid h'_C(M_{I_C}), h_C(M_{I_C})) \notag\\
&\leq H(M_{S \setminus I_C}) - H(M_{S \setminus I_C} \mid M_{I_C}) \notag\\
&= 0, \label{eq:lemma3}
\end{align}
where Equality \eqref{eq:lemma3} holds because \(M_{I_C}\) and \(M_{S \setminus I_C}\) are independent. Thus, \(H(h_C(M_{I_C}) \mid h'_C(M_{I_C})) = 0\). By symmetry, we also obtain \(H(h'_C(M_{I_C}) \mid h_C(M_{I_C})) = 0\).

We next prove that \(h_C \cap g \subseteq h'_C \cap g\). Consider any function \(d \in h_C \cap g\). By definition,
\[
H(d \mid h_C) = 0 \quad \text{and} \quad H(d \mid g) = 0.
\]
Since \(H(h_C \mid h'_C) = 0\), \(h_C\) is a deterministic function of \(h'_C\). Combining this with \(H(d \mid h_C) = 0\), it follows that \(d\) is also a deterministic function of \(h'_C\), implying \(H(d \mid h'_C) = 0\). Therefore, \(d \in h'_C \cap g\). By the arbitrariness of \(d\), we conclude \(h_C \cap g \subseteq h'_C \cap g\).

Similarly, we can derive \(h'_C \cap g \subseteq h_C \cap g\). Hence, \(h_C \cap g = h'_C \cap g\).
\end{proof}

From Lemma \ref{lem:strong decomposable}, the set of common functions \(h_C\cap g\)
does not depend on the choice of intermediate function \(h_C\), but only on the strongly decomposable pair \((C,h)\) and \(g\).
Furthermore, the entropy \(H(h_C\sqcap g)\) of the maximal common function
is also independent of the choice of \(h_C\),
and depends only on \((C,h)\) and \(g\).

\begin{definition}
For a cut set \(C\in\Lambda(\mathcal{N})\) and a wiretap set \(W\in\mathcal{W}_r\),
the pair \((C,W)\) is said to be \emph{valid} if the following hold:
\begin{itemize}
    \item $(C,f)$ is strongly decomposable;
    \item $W\subseteq C$ and $D_W\subseteq I_C$.
\end{itemize}
\end{definition}

\begin{theorem}\label{thm:the upper bound}
For the secure network function computation model \((\mathcal{N},f,\mathcal{W}_r,\zeta)\),
the secure computing capacity $\hat{\mathcal{C}}(\mathcal{N},f,\mathcal{W}_r,\zeta)$ satisfies
\[
\hat{\mathcal{C}}(\mathcal{N},f,\mathcal{W}_r,\zeta)
\leq
\min_{(C,W)\ \text{is valid}}
\frac{|C|-|W|}{H(f_C\sqcap\zeta)}\cdot\log|\mathcal{B}|,
\]
where \(f_C\) denotes any intermediate function of a strong decomposition of the strongly decomposable pair \((C,f)\).
\end{theorem}

\begin{proof}
Let \(\hat{C}\) be an arbitrary admissible \((\ell,n)\) function-computing secure network code for the model \((\mathcal{N},f,\mathcal{W}_r,\zeta)\), with global encoding functions \(\{\hat{g}_e : e \in \mathcal{E}\}\). Consider an arbitrary valid pair \((C,W)\), where $C \in \Lambda(\mathcal{N})$ and $W \in \mathcal{W}_r$. Since \((C,f)\) is strongly decomposable, there exist functions \(\hat{f}\) and \(f_C\) such that
\[
f(M_S) = \hat{f}\bigl(f_C(M_{I_C}), M_{S \setminus I_C}\bigr),
\]
where \(f_C\) is a function of \(M_{I_C}\), and \(\hat{f}(\,\cdot\,, M_{S \setminus I_C})\) is injective in its first argument for a fixed \(M_{S \setminus I_C}\).

We now consider the entropy of the maximal common function \((f_C \sqcap \zeta)(\mathbf{M}_S)\), where
\[
(f_C \sqcap \zeta)(\mathbf{M}_S) = \bigl((f_C \sqcap \zeta)(M_{1,j}, M_{2,j}, \dots, M_{s,j}) : j = 1, 2, \dots, \ell\bigr).
\]
From Lemma \ref{lem:strong decomposable}, we know that \(f_C \sqcap \zeta\) is independent of the choice of \(f_C\) and is fully determined by \(C\), \(f\), and \(\zeta\).

By the security condition \(I(\zeta(\mathbf{M}_S); Y_W) = 0\), we have
\begin{align}
I((f_C \sqcap \zeta)(\mathbf{M}_S); Y_W)
&= H(Y_W) - H(Y_W \mid (f_C \sqcap \zeta)(\mathbf{M}_S)) \notag\\
&\leq H(Y_W) - H(Y_W \mid \zeta(\mathbf{M}_S), (f_C \sqcap \zeta)(\mathbf{M}_S)) \notag\\
&= H(Y_W) - H(Y_W \mid \zeta(\mathbf{M}_S)) \label{eq:zeta1}\\
&= 0, \notag
\end{align}
where Equality \eqref{eq:zeta1} holds because \(H(f_C \sqcap \zeta \mid \zeta) = 0\).
This implies
\begin{align}\label{eq:bound5}
H((f_C \sqcap \zeta)(\mathbf{M}_S)) = H((f_C \sqcap \zeta)(\mathbf{M}_S) \mid Y_W).
\end{align}
Since \(H(f_C \sqcap \zeta \mid f_C) = 0\), \(f_C \sqcap \zeta\) is a deterministic function of \(f_C\), which confirms that \(f_C \sqcap \zeta\) is a function of \(\mathbf{M}_{I_C}\). Combining with \(D_W \subseteq I_C\), we know that \(Y_W\) is fully determined by \(\mathbf{M}_{I_C}\) and \(\mathbf{K}_{I_C}\). Thus, we have
\[
H((f_C \sqcap \zeta)(\mathbf{M}_S), Y_W) = H((f_C \sqcap \zeta)(\mathbf{M}_S), Y_W \mid \mathbf{M}_{S \setminus I_C}, \mathbf{K}_{S \setminus I_C}),
\]
which expands to
\[
H(Y_W) + H((f_C \sqcap \zeta)(\mathbf{M}_S) \mid Y_W) = H(Y_W \mid \mathbf{M}_{S \setminus I_C}, \mathbf{K}_{S \setminus I_C}) + H((f_C \sqcap \zeta)(\mathbf{M}_S) \mid Y_W, \mathbf{M}_{S \setminus I_C}, \mathbf{K}_{S \setminus I_C}).
\]
Since $H(Y_W) = H(Y_W \mid \mathbf{M}_{S \setminus I_C}, \mathbf{K}_{S \setminus I_C})$, it follows that
\begin{align}\label{eq:bound6}
H((f_C \sqcap \zeta)(\mathbf{M}_S) \mid Y_W) = H((f_C \sqcap \zeta)(\mathbf{M}_S) \mid Y_W, \mathbf{M}_{S \setminus I_C}, \mathbf{K}_{S \setminus I_C}).
\end{align}

Next, define the edge subset $C' \triangleq \bigcup_{\sigma \in S \setminus I_C} \text{Out}(\sigma)$ and let $\widetilde{C} = C \cup C'$. Evidently, $\widetilde{C}$ is a global cut set, meaning it separates $\rho$ from all source nodes in $S$. Combined with the acyclicity of the graph $\mathcal{G}$, this implies that $Y_{\text{In}(\rho)}$ is a function of $Y_{\widetilde{C}}$. Given the admissibility of the secure network code $\hat{C}$, the target function $f$ must be computable with zero error at $\rho$, i.e.,
\[
H\bigl(f(\mathbf{M}_S) \mid Y_{\text{In}(\rho)}\bigr) = 0,
\]
which implies
\begin{align*}
H\bigl(f(\mathbf{M}_S) \mid Y_{\widetilde{C}}\bigr)
&= H\bigl(f(\mathbf{M}_S) \mid Y_{\widetilde{C}}, Y_{\text{In}(\rho)}\bigr) \\
&\leq H\bigl(f(\mathbf{M}_S) \mid Y_{\text{In}(\rho)}\bigr) = 0.
\end{align*}
Since $D_{C'} = I_{C'} = S \setminus I_C$, $Y_{C'}$ is a function of $(\mathbf{M}_{S \setminus I_C}, \mathbf{K}_{S \setminus I_C})$, and thus
\begin{align}\label{eq:bound1}
H\bigl(f(\mathbf{M}_S) \mid Y_C, \mathbf{M}_{S \setminus I_C}, \mathbf{K}_{S \setminus I_C}\bigr) = 0.
\end{align}
Since $\hat{f}(\,\cdot\,, M_{S \setminus I_C})$ is injective in its first argument for a fixed $M_{S \setminus I_C}$, we have
\begin{align}\label{eq:bound2}
H\bigl(f_C(\mathbf{M}_{I_C}) \mid f(\mathbf{M}_S), \mathbf{M}_{S \setminus I_C}\bigr) = 0.
\end{align}
By the definition of the maximal common function $f_C \sqcap \zeta$, we have
\begin{align}\label{eq:bound3}
H(f_C \sqcap \zeta \mid f_C) = 0.
\end{align}
Combining \eqref{eq:bound1}, \eqref{eq:bound2}, and \eqref{eq:bound3}, we obtain
\begin{align}\label{eq:bound4}
H\bigl((f_C \sqcap \zeta)(\mathbf{M}_S) \mid Y_C, \mathbf{M}_{S \setminus I_C}, \mathbf{K}_{S \setminus I_C}\bigr) = 0.
\end{align}
Combining \eqref{eq:bound5}, \eqref{eq:bound6}, and \eqref{eq:bound4}, we have
\begin{align}
\ell \cdot H((f_C \sqcap \zeta)(M_S)) = H((f_C \sqcap \zeta)(\mathbf{M}_S))
&= H((f_C \sqcap \zeta)(\mathbf{M}_S) \mid Y_W, \mathbf{M}_{S \setminus I_C}, \mathbf{K}_{S \setminus I_C}) \notag\\
&\quad - H((f_C \sqcap \zeta)(\mathbf{M}_S) \mid Y_C, \mathbf{M}_{S \setminus I_C}, \mathbf{K}_{S \setminus I_C}) \notag\\
&= I\bigl((f_C \sqcap \zeta)(\mathbf{M}_S); Y_{C \setminus W} \mid Y_W, \mathbf{M}_{S \setminus I_C}, \mathbf{K}_{S \setminus I_C}\bigr) \notag\\
&\leq H\bigl(Y_{C \setminus W} \mid Y_W, \mathbf{M}_{S \setminus I_C}, \mathbf{K}_{S \setminus I_C}\bigr) \notag\\
&\leq (|C| - |W|) \cdot n \cdot \log|\mathcal{B}|. \notag
\end{align}
Rearranging yields the upper bound:
\begin{align}\label{eq:bound}
\frac{\ell}{n} \leq \frac{|C| - |W|}{H(f_C \sqcap \zeta)} \cdot \log|\mathcal{B}|.
\end{align}
Since Inequality \eqref{eq:bound} holds for all valid pairs \((C,W)\), we have
\[
\frac{\ell}{n} \leq \min_{(C,W)\ \text{is valid}} \frac{|C| - |W|}{H(f_C \sqcap \zeta)} \cdot \log|\mathcal{B}|.
\]
Finally, since this upper bound is valid for all admissible secure network codes for the model $(\mathcal{N}, f, \mathcal{W}_r, \zeta)$, the theorem is proven.
\end{proof}

We next consider a case where the secure computing capacity can be exactly characterized by the upper bound in Theorem \ref{thm:the upper bound}.
\begin{corollary}
Consider the model $(\mathcal{N}, f, \mathcal{W}_r, \zeta)$. If the security level \(r\) satisfies
\[
r \geq \overline{C}_{\min,f} \triangleq \min \{|C|: C \in \Lambda(\mathcal{N}), D_{C}=I_{C}\ \text{and}\ (C,f) \ \text{is strongly decomposable} \}
\]
then
\[
\hat{\mathcal{C}}(\mathcal{N}, f, \mathcal{W}_r, \zeta) = 0.
\]
\end{corollary}

\begin{proof}
Let \(C\) be a cut set in \(\mathcal{N}\) satisfying $D_C = I_C$, $(C,f)$ being strongly decomposable, and $|C| = \overline{C}_{\min,f}$.
Since \(\overline{C}_{\min,f} \leq r\), we have $|C| \leq r$, which implies $C \in \mathcal{W}_r$.
Combined with $D_C = I_C$, we know that $(C,C)$ is a valid pair. Substituting this into the upper bound in Theorem \ref{thm:the upper bound} yields
\(
\hat{\mathcal{C}}(\mathcal{N}, f, \mathcal{W}_r, \zeta) \leq 0
\). Since the secure computing capacity is non-negative, we conclude
\(
\hat{\mathcal{C}}(\mathcal{N}, f, \mathcal{W}_r, \zeta) = 0
\).
\end{proof}

\begin{corollary}\label{cor:upper upper}
If the security function $\zeta$ is the \emph{identity function}, i.e., $\zeta(M_S) = M_S$, and for all $S'\subseteq S$, the target function $f$ can be decomposed as
\[
f(M_S) = \hat{f}\bigl(f'(M_{S'}), M_{S \setminus S'}\bigr),
\]
where $\hat{f}(\,\cdot\,, M_{S \setminus S'})$ is injective in its first argument for any fixed $M_{S \setminus S'}$, then the upper bound on the secure computing capacity $\hat{\mathcal{C}}(\mathcal{N}, f, \mathcal{W}_r, \zeta)$ obtained in Theorem \ref{thm:the upper bound} satisfies
\begin{align}\label{eq:the upper bound1}
\min_{(C,W)\ \text{is valid}} \frac{|C| - |W|}{H(f_C \sqcap \zeta)} \cdot \log|\mathcal{B}| \leq \min_{C \in \Lambda(\mathcal{N})} \frac{|C|}{H(f(M_{I_C}, \mathbf{0}))} \cdot \log|\mathcal{B}|,
\end{align}
and
\begin{align}\label{eq:the upper bound 2}
\min_{(C,W)\ \text{is valid}} \frac{|C| - |W|}{H(f_C \sqcap \zeta)} \cdot \log|\mathcal{B}| \geq \min_{C \in \Lambda(\mathcal{N})} \frac{|C| - r}{H(f(M_{I_C}, \mathbf{0}))} \cdot \log|\mathcal{B}|.
\end{align}
\end{corollary}

\begin{proof}
We first prove the upper bound in \eqref{eq:the upper bound1}. Since $\zeta$ is the identity function, we have $H(f_C \mid \zeta) = 0$. This implies $f_C \in f_C \cap \zeta$. For any $d \in f_C \cap \zeta$, we have $H(d \mid f_C) = 0$, which further leads to $H(d) \leq H(f_C)$. Therefore, $f_C \sqcap \zeta = f_C$.

By the conditions satisfied by the target function $f$, for all $C \in \Lambda(\mathcal{N})$, the pair $(C,f)$ is strongly decomposable. Therefore, $(C,W)$ is valid if and only if $W \subseteq C$ and $D_W \subseteq I_C$. Hence,
\begin{align}
\min_{(C,W)\ \text{is valid}} \frac{|C| - |W|}{H(f_C \sqcap \zeta)} \cdot \log|\mathcal{B}|
&= \min_{W\in \mathcal{W}_{r}}\min_{\substack{C\in\Lambda(\mathcal{N})\\ W\subseteq C\ \text{and}\ D_{W}\subseteq I_{C}}} \frac{|C| - |W|}{H(f_C)} \cdot \log|\mathcal{B}|\notag\\
&\leq \min_{C \in \Lambda(\mathcal{N})} \frac{|C|}{H(f_C)} \cdot \log|\mathcal{B}|, \label{eq:empty}
\end{align}
where inequality \eqref{eq:empty} holds because the empty set $\emptyset \in \mathcal{W}_r$ is a valid wiretap set, satisfying $\emptyset \subseteq C$ and $D_{\emptyset} \subseteq I_C$ for any cut set $C \in \Lambda(\mathcal{N})$.

Next, we analyze $H(f_C(M_{I_C}))$ using the strong decomposition property. For any $m = M_{S \setminus I_C}$, the strong decomposition gives
\[
f(M_{I_C}, m) = \hat{f}\bigl(f_C(M_{I_C}), m\bigr).
\]
Since $f(M_{I_C}, m)$ is a deterministic function of $f_C(M_{I_C})$, we have
\[
H\bigl(f(M_{I_C}, m) \mid f_C(M_{I_C})\bigr) = 0,
\]
which implies $H\bigl(f(M_{I_C}, m)\bigr) \leq H\bigl(f_C(M_{I_C})\bigr)$. On the other hand, since $\hat{f}(\,\cdot\,, m)$ is injective, it follows that
\[
H\bigl(f_C(M_{I_C}) \mid f(M_{I_C}, m)\bigr) = 0,
\]
which implies $H\bigl(f_C(M_{I_C})\bigr) \leq H\bigl(f(M_{I_C}, m)\bigr)$. Therefore, for any $m = M_{S \setminus I_C}$, we have
\[
H\bigl(f_C(M_{I_C})\bigr) = H\bigl(f(M_{I_C}, m)\bigr).
\]
In particular, by setting $m = \mathbf{0}$, we obtain
\begin{align}\label{eq:0}
H\bigl(f_C(M_{I_C})\bigr) = H\bigl(f(M_{I_C}, \mathbf{0})\bigr).
\end{align}

Combining \eqref{eq:empty} and \eqref{eq:0}, we conclude
\[
\min_{(C,W)\ \text{is valid}} \frac{|C| - |W|}{H(f_C \sqcap \zeta)} \cdot \log|\mathcal{B}|
\leq \min_{C \in \Lambda(\mathcal{N})} \frac{|C|}{H(f(M_{I_C}, \mathbf{0}))} \cdot \log|\mathcal{B}|.
\]

Next, we prove the lower bound in \eqref{eq:the upper bound 2}. Since $W \in \mathcal{W}_r$, we have $|W| \leq r$. Thus,
\begin{align}
\min_{(C,W)\ \text{is valid}}  \frac{|C| - |W|}{H(f_C \sqcap \zeta)} \cdot \log|\mathcal{B}|
&= \min_{C\in\Lambda(\mathcal{N})}\min_{\substack{W\in \mathcal{W}_{r}\\ W\subseteq C\ \text{and}\ D_{W}\subseteq I_{C}}} \frac{|C| - |W|}{H(f_C)} \cdot \log|\mathcal{B}|\notag\\
&\geq \min_{C \in \Lambda(\mathcal{N})} \frac{|C| - r}{H(f_C)} \cdot \log|\mathcal{B}|\notag\\
&= \min_{C \in \Lambda(\mathcal{N})} \frac{|C| - r}{H(f(M_{I_C}, \mathbf{0}))} \cdot \log|\mathcal{B}|,
\end{align}
where the last equality follows from \eqref{eq:0}.
\end{proof}

\begin{corollary}\label{cor:spec}
If the target function $f$ is a scalar linear function over a finite field $\mathbb{F}_q$, i.e.,
\[
f(m_1, m_2, \dots, m_s) = \sum_{i=1}^s \alpha_i m_i, \quad \alpha_i \in \mathbb{F}_q^{*},\ i \in [s],
\]
and the security function $\zeta$ is the identity function, then
\[
\hat{\mathcal{C}}(\mathcal{N}, f, \mathcal{W}_r, \zeta) \leq \min_{\substack{(W,C) \in \mathcal{W}_r \times \Lambda(\mathcal{N}) \\ W \subseteq C,\ D_W \subseteq I_C}} (|C| - |W|)\cdot \log_{q}|\mathcal{B}|.
\]
\end{corollary}

\begin{proof}
For any subset $S' \subseteq S$, the function $f$ admits the decomposition
\[
f(m_1, m_2, \dots, m_s) = \sum_{i \in [s]: \sigma_{i}\in S'} \alpha_i m_i + \sum_{i \in [s]: \sigma_{i}\in  S \setminus S'} \alpha_i m_i.
\]
Fix $m_{S \setminus S'}=\{m_{i}:\sigma_{i}\in S \setminus S' \}$. If $m_{S'}=\{m_{i}:\sigma_{i}\in S' \}$ and $m'_{S'} =\{m'_{i}:\sigma_{i}\in S' \}$ satisfy
\[
\sum_{i \in [s]: \sigma_{i}\in S'} \alpha_i m_i \neq \sum_{i \in [s]: \sigma_{i}\in S'} \alpha_i m_i',
\]
then $f(m_{S'}, m_{S \setminus S'}) \neq f(m_{S'}', m_{S \setminus S'})$. It follows that $f$ and $\zeta$ satisfy the conditions in Corollary \ref{cor:upper upper}, which implies
\[
\hat{\mathcal{C}}(\mathcal{N}, f, \mathcal{W}_r, \zeta) \leq  \min_{\substack{(W,C) \in \mathcal{W}_r \times \Lambda(\mathcal{N}) \\ W \subseteq C,\ D_W \subseteq I_C}} \frac{|C| - |W|}{H(f_{C})} \log|\mathcal{B}|.
\]

We may take $f_C(m_{I_C}) = \sum_{i \in [s]: \sigma_{i}\in I_C} \alpha_i m_i$. Since $I_C \neq \emptyset$, the entropy satisfies $H(f_C) = \log q$. Substituting this into the inequality above yields
\[
\hat{\mathcal{C}}(\mathcal{N}, f, \mathcal{W}_r, \zeta) \leq \min_{\substack{(W,C) \in \mathcal{W}_r \times \Lambda(\mathcal{N}) \\ W \subseteq C,\ D_W \subseteq I_C}} (|C| - |W|)\cdot \log_{q}|\mathcal{B}|.
\]
\end{proof}

Note that Corollary \ref{cor:spec} shows that in this specific scenario, our proposed upper bound in Theorem \ref{thm:the upper bound} degenerates into the upper bound previously established in {\cite[Theorem 1]{GBY2024}}. This further confirms that Theorem \ref{thm:the upper bound} is a generalization of the existing upper bound in \cite{GBY2024}, as our conclusion applies to arbitrary target functions and security functions.

\begin{theorem}
\label{thm:upper_bound_linear}
For the secure network function computation model $(\mathcal{N},f,\mathcal{W}_r,\zeta)$,
the secure computing capacity $\hat{\mathcal{C}}(\mathcal{N},f,\mathcal{W}_r,\zeta)$ satisfies
\[
\hat{\mathcal{C}}(\mathcal{N},f,\mathcal{W}_r,\zeta)
\leq
\min_{C \in \Lambda(\mathcal{N})} \frac{|C|}{H(f(M_{I_C}, \mathbf{0}))} \cdot \log |\mathcal{B}|.
\]
\end{theorem}

\begin{proof}
Consider any admissible $(\ell, n)$ function-computing secure network code $\mathcal{C}$ for the model $(\mathcal{N}, f, \mathcal{W}_r, \zeta)$.
By the computability condition of $\mathcal{C}$, we have
\[
H(f(\mathbf{M}_{S}) \mid Y_{\operatorname{In}(\rho)}) = 0,
\]
where $Y_{\operatorname{In}(\rho)}$ denotes the random variables observed by $\rho$.

For any $C \in \Lambda(\mathcal{N})$, we have
$H(Y_{\operatorname{In}(\rho)} \mid Y_C, \mathbf{M}_{S \setminus I_C}, \mathbf{K}_{S \setminus I_C}) = 0$.
Fixing $\mathbf{M}_{S \setminus I_C} = \mathbf{0}$ and $\mathbf{K}_{S \setminus I_C} = \mathbf{0}$, we obtain
\[
H(f(\mathbf{M}_{I_{C}},\mathbf{0}) \mid Y_C) = 0.
\]
This further implies
\[
\ell\cdot H(f(M_{I_C}, \mathbf{0}))=H(f(\mathbf{M}_{I_C}, \mathbf{0}))\leq  H(Y_C)\leq n|C|\cdot \log|\mathcal{B}|.
\]
Thus,
\[
\frac{\ell}{n}\leq \frac{|C|}{H(f(M_{I_C}, \mathbf{0}))} \cdot \log |\mathcal{B}|.
\]
Taking the minimum over all $C \in \Lambda(\mathcal{N})$ completes the proof.
\end{proof}

\begin{remark}
When $f$ and $\zeta$ satisfy the conditions of Corollary \ref{cor:upper upper}, the upper bound on the secure computation capacity provided by Theorem \ref{thm:upper_bound_linear} is not tighter than that in Theorem \ref{thm:the upper bound}. However, for other configurations of $f$ and $\zeta$, the upper bound proposed in Theorem \ref{thm:upper_bound_linear} may yield a better (tighter) bound compared to Theorem \ref{thm:the upper bound}. Hence, to obtain a generally tight upper bound on the secure computation capacity, it is necessary to consider the conclusions of both Theorem \ref{thm:the upper bound} and Theorem \ref{thm:upper_bound_linear} in combination.
\end{remark}

\begin{corollary}\label{cor:spec2}
If both the target function \(f\) and the security function \(\zeta\) are scalar linear functions over $\mathbb{F}_q$, i.e.,
\[
f(m_S) =\zeta(m_S) =  \sum_{i =1}^{s} \alpha_i m_i, \quad \alpha_i \in \mathbb{F}_q^{*},\ \forall i \in [s],
\]
then
\[
\hat{\mathcal{C}}(\mathcal{N}, f, \mathcal{W}_r, \zeta) \leq \min_{\substack{(W,C) \in \mathcal{W}_r \times \Lambda(\mathcal{N}): \\ W \subseteq C, \ I_C=D_{W}=S\ \text{or}\\ W \subseteq C, \ I_C\setminus D_{W}\neq\emptyset}} (|C| - |W|)\cdot \log_{q}|\mathcal{B}|.
\]
\end{corollary}

\begin{proof}
For any $C \in \Lambda(\mathcal{N})$, the function $f$ admits the decomposition
\[
f(m_S) = \sum_{i \in [s]: \sigma_{i}\in I_{C}} \alpha_i m_i + \sum_{i \in [s]: \sigma_{i}\in S\setminus I_{C}} \alpha_i m_i.
\]
Thus, $(C,f)$ is always strongly decomposable, and $(C,W)$ is valid if and only if $W \subseteq C$ and $D_W \subseteq I_C$.

When $I_C = S$, we have $f_C = f$. Combining $f = \zeta$, we obtain $f_C \sqcap \zeta = f$. Since $H(f) = \log q$, it follows from Theorem \ref{thm:the upper bound} that
\begin{align}\label{eq:theorem1}
\hat{\mathcal{C}}(\mathcal{N}, f, \mathcal{W}_r, \zeta) \leq \min_{\substack{(W,C) \in \mathcal{W}_r \times \Lambda(\mathcal{N}) \\ W \subseteq C,\ I_C=S}} (|C| - |W|) \log_{q} |\mathcal{B}|.
\end{align}

On the other hand, for any $C \in \Lambda(\mathcal{N})$, we have $I_C \neq \emptyset$. Thus $H(f(M_{I_C},\mathbf{0})) = \log q$. From Theorem \ref{thm:upper_bound_linear}, we obtain
\begin{align}\label{eq:theorem2}
\hat{\mathcal{C}}(\mathcal{N}, f, \mathcal{W}_r, \zeta) \leq \min_{C \in \Lambda(\mathcal{N})} |C|\cdot \log_{q}|\mathcal{B}|.
\end{align}

Combining inequalities \eqref{eq:theorem1} and \eqref{eq:theorem2} completes the proof.
\end{proof}

Note that Corollary \ref{cor:spec2} demonstrates that, in this specific scenario, our proposed upper bound on the secure computing capacity reduces to the upper bound established in {\cite[Theorem 1]{arxiv2025secure}}.

\begin{figure}[htbp]
    \centering
\begin{tikzpicture}[
    source/.style={circle,draw,minimum size=6pt,inner sep=1pt},
    every edge/.style={draw, -Stealth, thick},
    middle/.style={circle,draw,minimum size=4pt, inner sep=3pt},
]
\node[source] (s1) at (0, 6) {$\sigma_1$};
\node[above=0 of s1] {$(m_1, k_1)$};

\node[source] (s2) at (4, 6) {$\sigma_2$};
\node[above=0 of s2] {$(m_2, k_2)$};

\node[middle] (m1) at (2, 4.5) {};
\node[middle] (m2) at (2, 2) {};

\node[middle] (b1) at (0, 0) {};
\node[middle] (b2) at (4, 0) {};
\node[source] (rho) at (2, -1.5) {$\rho$};
\node[below=0 of rho] {$m_{1}m_{2}$};

\path (s1) edge node[below,left] {$k_1$} (b1) ;
\path (s1) edge node[sloped,above] {$m_{1}k_{1}$} (m1);

\path (s2) edge node[sloped,above] {$m_{2}k_{2}$} (m1);
\path (s2) edge node[midway,right] {$k_2$} (b2);

\path (m1) edge node[sloped,above] {$m_{1}m_{2}k_{1}k_{2}$} (m2) ;

\path (m2) edge node[sloped,above] {$m_{1}m_{2}k_{1}k_{2}$} (b1);
\path (m2) edge node[sloped,above] {$m_{1}m_{2}k_{1}k_{2}$} (b2);

\path (b1) edge node[sloped,above] {$m_{1}m_{2}k_{2}$} (rho);
\path (b2) edge node[sloped,above] {$k_{2}$} (rho);
\end{tikzpicture}
\caption{An admissible $(1, 1)$ secure network code for the secure model $(\mathcal{N}, f, \mathcal{W}_{r},\zeta)$}
    \label{fig:base_code_f1}
\end{figure}
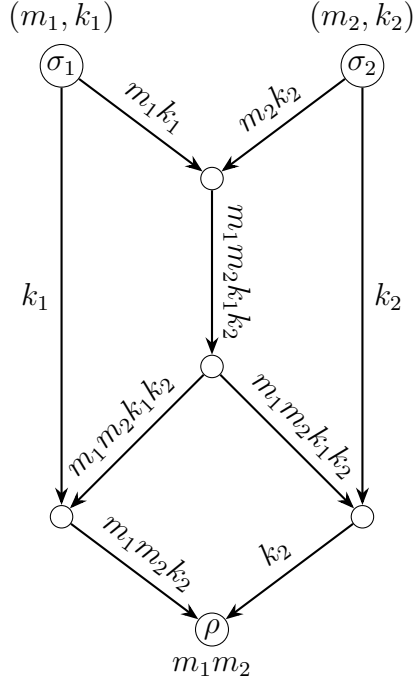

\begin{example}
We consider a secure model $(\mathcal{N}, f, \mathcal{W}_r, \zeta)$, where $\mathcal{N} = (\mathcal{G}, S, \rho)$ is the reverse butterfly network depicted in \autoref{fig:base_code_f1} with the set of source nodes $S = \{\sigma_1, \sigma_2\}$. Let $\mathcal{B} = \mathbb{F}_3^*$, meaning that each edge in the network can reliably transmit one symbol from $\mathbb{F}_3^*$ per use. The target function $f$ is multiplication over the finite field $\mathbb{F}_3$, i.e.,
\[f(m_1, m_2) = m_1 m_2 \pmod{3}, \quad \text{for } m_i \in \mathbb{F}_3^*, i=1,2.\]
The security function $\zeta$ is the identity function, and the security level is $r=1$. By Theorem \ref{thm:the upper bound}, since for any $C \in \Lambda(\mathcal{N})$, the pair $(C, f)$ is strongly decomposable, and
\[H(M_1) = H(M_2) = H(M_1 M_2) = 1,\]
we immediately obtain
\[
\hat{\mathcal{C}}(\mathcal{N}, f, \mathcal{W}_r, \zeta) \leq 1.
\]

On the other hand, we present an admissible $(1,1)$ secure network code as illustrated in \autoref{fig:base_code_f1}, where $m_i, k_i \in \mathbb{F}_3^*$ denote arbitrary realizations of the source message $\mathbf{M}_i$ and the random key $\mathbf{K}_i$ at source node $\sigma_i$, respectively. Hence,
\[
\hat{\mathcal{C}}(\mathcal{N}, f, \mathcal{W}_r, \zeta) = 1.
\]
\end{example}

\section{An Efficient Algorithm for the Upper Bound on Secure Computing Capacity for Linear Target Functions under Source Security}\label{sec:algorithm}

In this section, we focus on a specific case of the secure network function computation model $(\mathcal{N}, f, \mathcal{W}_r, \zeta)$, where $f$ and $\zeta$ are explicitly specified as follows. The security function $\zeta$ is the identity function, meaning our goal is to prevent source information from leaking to wiretappers. The target function $f$ is a linear function over $\mathbb{F}_q$, a fundamental class of functions that has been extensively studied in the literature. Specifically,
\[ f(m_1,m_2,\dots,m_s) = \left(\sum_{i=1}^s a_{i,j} m_i : j=1,2,\dots,k\right),\]
where $m_i, a_{i,j} \in \mathbb{F}_q$ for all $i = 1,2,\dots,s$ and $j = 1,2,\dots,k$. This implies $\mathcal{A}_i = \mathbb{F}_q$ for all $1 \leq i \leq s$ and $\mathcal{O} = \mathbb{F}_q^k$. For notational convenience, we rewrite $f$ in matrix form as
\begin{align}\label{eq:the target function}
f(\vec{m}_S) = \vec{m}_S \cdot T,
\end{align}
where $\vec{m}_S = (m_1,m_2,\dots,m_s)$ and $T \triangleq (a_{i,j})_{1 \leq i \leq s, 1 \leq j \leq k}$ is the coefficient matrix of $f$. Clearly, the linear function $f$ is completely determined by $T$, and we henceforth use $T$ to represent $f$. Without loss of generality, we assume $T$ has full column rank, which implies $k \leq s$. Correspondingly, all independent source messages $M_i$ ($i = 1,2,\dots,s$) are uniformly distributed over $\mathbb{F}_q$.

For brevity, we abbreviate the secure model $(\mathcal{N}, f, \mathcal{W}_r, \zeta)$ and the secure computing capacity $\hat{\mathcal{C}}(\mathcal{N}, f, \mathcal{W}_r,\zeta)$ as $(\mathcal{N}, T, r)$ and $\hat{\mathcal{C}}(\mathcal{N}, T, r)$, respectively. We assume without loss of generality that $\mathcal{B} = \mathbb{F}_q$, meaning each edge can reliably transmit an element of $\mathbb{F}_q$ per network use.

We assume $T$ has no all-zero rows. If there exists an index $i\in[s]$ such that $a_{ij}=0$ for all $1\leq j \leq k$, we remove the corresponding row from $T$ to form a new coefficient matrix $T'$. Simultaneously, we remove the source node $\sigma_{i}$ along with all its outgoing edges $\operatorname{Out}(\sigma_{i})$ from the graph $\mathcal{G}$ to construct a new graph $\mathcal{G}'$. We then update the network $\mathcal{N}$ to
\[\mathcal{N}'=\left(\mathcal{G}', S \backslash\left\{\sigma_{i}: a_{ij}=0\ \text{for all } 1\leq j \leq k\right\}, \rho\right).\]
Consider a new target function $f'$, a linear function over $\mathbb{F}_{q}$ with coefficient matrix $T'$. It is straightforward to verify that securely computing $f$ (with coefficient matrix $T$) over $\mathcal{N}$ at security level $r$ is equivalent to securely computing $f'$ (with coefficient matrix $T'$) over $\mathcal{N}'$ at the same security level, and thus $\hat{\mathcal{C}}(\mathcal{N}, T, r)=\hat{\mathcal{C}}(\mathcal{N}', T', r)$. Therefore, we assume without loss of generality that every row of $T$ is non-zero in the model $(\mathcal{N}, T, r)$.

For a subset of source nodes $A \subseteq S$, let $T_{A}$ denote the submatrix of $T$ consisting of all rows corresponding to the source nodes in $A$. From Theorem \ref{thm:the upper bound}, we immediately obtain an upper bound on $\hat{\mathcal{C}}(\mathcal{N}, T, r)$.

\begin{theorem}\label{cor:the upper bound}
Consider the secure network function computation model $(\mathcal{N}, T, r)$, where the target function $f$ is a linear function over $\mathbb{F}_q$ with coefficient matrix $T\in \mathbb{F}_{q}^{s\times k}$. Then,
\begin{align}
\hat{\mathcal{C}}(\mathcal{N}, T, r) \leq \min_{\substack{(W, C) \in \mathcal{W}_{r} \times \Lambda(\mathcal{N}): \\ W \subseteq C\ \text{and} \ D_{W} \subseteq I_{C}}}\frac{|C|-|W|}{\operatorname{Rank}(T_{I_{C}})}.
\end{align}
\end{theorem}

\begin{proof}
For any \( C \in \Lambda(\mathcal{N}) \), the function \( f \) can be decomposed as
\[
f(\vec{m}_S) = \vec{m}_{I_C} \cdot T_{I_C} + \vec{m}_{S \setminus I_C} \cdot T_{S \setminus I_C},
\]
where \( \vec{m}_A = (m_i : \sigma_{i} \in A) \) for any \( A \subseteq S \). Thus, \( (C,f) \) is strongly decomposable for all \( C \in \Lambda(\mathcal{N}) \), and \( (C,W) \) is valid if and only if \( W \subseteq C \) and \( D_W \subseteq I_C \). In addition, we have $f_{C}(\vec{m}_{I_C})=\vec{m}_{I_C}\cdot T_{I_C}$. Since \( \zeta \) is the identity function, we have \( f_C \sqcap \zeta = f_C \).

We now compute the entropy \( H((f_C \sqcap \zeta)(M_{I_C})) = H(M_{I_C} \cdot T_{I_C}) \). Since \( I_C \neq \emptyset \) and every row of \( T \) is non-zero, we have \( \operatorname{Rank}(T_{I_C}) \geq 1 \). There exists an \( |I_C| \times \operatorname{Rank}(T_{I_C}) \) submatrix $T'_{I_C}$ of $T_{I_C}$ with full column rank, whose column span is identical to that of $T_{I_C}$. Let $T_{I_{C}}^c$ denote the submatrix of $T_{I_{C}}$ obtained by removing $T'_{I_C}$. Since every column of $T_{I_{C}}^c$ is a linear combination of the columns of $T'_{I_C}$, we have
\[
H( M_{I_C} \cdot T_{I_{C}}^c \mid M_{I_C} \cdot T'_{I_C} )=0.
\]
Thus
\begin{align}
H ( M_{I_C} \cdot T_{I_{C}})
=& H(M_{I_C} \cdot T'_{I_C})+H( M_{I_C} \cdot T_{I_{C}}^c \mid M_{I_C} \cdot T'_{I_C}) \notag \\
=& H(M_{I_C} \cdot T'_{I_C}). \label{eq:xiangdeng}
\end{align}
Recall that all source messages $M_{i}$ ($i=1,2,\dots,s$) are independent and uniformly distributed over $\mathbb{F}_q$. Thus, $M_{I_{C}}$ is uniformly distributed over $\mathbb{F}_{q}^{|I_{C}|}$, i.e., for any $\mathbf{m}\in \mathbb{F}_{q}^{|I_{C}|}$,
\begin{align*}
\operatorname{Pr}\left(M_{I_{C}}=\mathbf{m}\right)=\frac{1}{q^{|I_{C}|}}.
\end{align*}
For any $\mathbf{y} \in \mathbb{F}_{q}^{\operatorname{Rank}(T_{I_{C}})}$, we have
\begin{align*}
\operatorname{Pr}\left(M_{I_{C}} \cdot T'_{I_C}=\mathbf{y}\right)
=&\sum_{\mathbf{m}\in \mathbb{F}_{q}^{|I_{C}|}}\operatorname{Pr}\left(M_{I_{C}}=\mathbf{m}\right)\cdot \operatorname{Pr}\left(M_{I_{C}} \cdot T'_{I_C}=\mathbf{y} \mid M_{I_{C}}=\mathbf{m}\right) \notag \\
=&\sum_{\substack{\mathbf{m}\in \mathbb{F}_{q}^{|I_{C}|}:\\\mathbf{m}\cdot T'_{I_C}=\mathbf{y}}} \operatorname{Pr}\left(M_{I_{C}}=\mathbf{m}\right) \notag\\
=& q^{|I_{C}|-\operatorname{Rank}(T'_{I_C})}\cdot \frac{1}{q^{|I_{C}|}} = \frac{1}{q^{\operatorname{Rank}(T_{I_{C}})}}.
\end{align*}
This implies $M_{I_{C}} \cdot T'_{I_C}$ is uniformly distributed over $\mathbb{F}_q^{\operatorname{Rank}(T_{I_{C}})}$. Together with \eqref{eq:xiangdeng}, we have
\begin{align*}
H(M_{I_{C}} \cdot T_{I_{C}})=\operatorname{Rank}(T_{I_{C}})\cdot \log{q}.
\end{align*}

From Theorem \ref{thm:the upper bound}, we obtain
\[
\hat{\mathcal{C}}(\mathcal{N}, T, r) \leq \min_{\substack{(W, C) \in \mathcal{W}_{r} \times \Lambda(\mathcal{N}): \\ W \subseteq C\ \text{and} \ D_{W} \subseteq I_{C}}}\frac{|C|-|W|}{\operatorname{Rank}(T_{I_{C}})}.
\]
\end{proof}

Note that when the security level $r=0$, Theorem \ref{cor:the upper bound} immediately yields
\[\begin{aligned}
\hat{\mathcal{C}}(\mathcal{N}, T, 0) & \leq \min_{C\in \Lambda(\mathcal{N})}\frac{|C|}{\operatorname{Rank}(T_{I_{C}})}.
\end{aligned}\]
The model of network function computation without information-theoretic security is a special case of the secure network function computation model. It can be observed that the upper bound in Theorem \ref{cor:the upper bound} matches the one presented in \cite{WXG2023}, which is the best known upper bound for computing linear target functions over arbitrary networks. Thus, Theorem \ref{cor:the upper bound} can be regarded as an extension to the secure setting.

We propose an efficient algorithm to compute the upper bound in Theorem \ref{cor:the upper bound}. We first investigate edge-set semilattices in the network and establish their key properties, which are crucial for developing an efficient algorithm to calculate the bound.

For a wiretap set $W \in \mathcal{W}_{r}$, we define
\begin{align}\label{eq:Omega}
\Omega(W) = \min_{\substack{C \in \Lambda(\mathcal{N}): \\ W \subseteq C \ \text{and}\  D_{W} \subseteq I_{C}}}\frac{|C|-|W|}{\operatorname{Rank}(T_{I_{C}})}.
\end{align}
Accordingly, the upper bound in Theorem \ref{cor:the upper bound} can be rephrased as
\[\hat{\mathcal{C}}(\mathcal{N}, T, r) \leq \min_{W \in \mathcal{W}_{r}} \Omega(W).\]

\begin{lemma}[{\cite[Lemma 5]{GBY2024}}]
Let $W$ be an edge subset and $W'$ be a minimum cut separating $W$ from $D_{W}$. Then \(D_{W'} = D_{W}\).
\end{lemma}

\begin{lemma}[{\cite[Lemma 6]{GBY2024}}] \label{lem:relation of omega}
Let $W \in \mathcal{W}_{r}$ be a wiretap set, and let $W'$ be a minimum cut separating $W$ from $D_{W}$. Then
\[\Omega\left(W'\right) \leq \Omega(W).\]
\end{lemma}

Note that if $W'$ is a minimum cut separating $W$ from $D_{W}$, then $W'$ is also a minimum cut separating $W'$ from $D_{W'}$, which implies $|W'| = \operatorname{mincut}(D_{W'},W')$. Based on this observation, we define the set $\mathcal{W}_{r}^{*} = \{W\subseteq \mathcal{E}: |W| = \operatorname{mincut}(D_{W},W) \leq r\}$. Consequently, the upper bound in Theorem \ref{cor:the upper bound} simplifies to
\[\hat{\mathcal{C}}(\mathcal{N}, T, r) \leq \min_{W \in \mathcal{W}_{r}^{*}} \Omega(W).\]

To facilitate subsequent analysis, we first introduce a natural partial order ``$\prec$'' on the edge set $\mathcal{E}$: for any two edges $d, e \in \mathcal{E}$, we write $d \prec e$ if there exists a directed path from $d$ to $e$ in $\mathcal{G}$. Building on the concept of minimum cuts, we further define a binary relation ``$\sim^{s}$'' on edge subsets: for any two edge subsets $\alpha, \beta \subseteq \mathcal{E}$, $\alpha \sim^{s} \beta$ if and only if there exists an edge subset $\xi$ that is a common minimum cut separating $\alpha$ from $D_{\alpha}$ and $\beta$ from $D_{\beta}$.

\begin{lemma}
The binary relation ``$\sim^{s}$'' is an equivalence relation on the edge subsets of $\mathcal{G}$. Specifically, for arbitrary edge subsets $W, W', W'' \subseteq \mathcal{E}$, it satisfies:
\begin{itemize}
\item (Reflexivity) $W \sim^{s} W$;
\item (Symmetry) If $W\sim^{s} W'$, then $W'\sim^{s} W$;
\item (Transitivity) If $W \sim^{s} W'$ and $W'\sim^{s} W''$, then $W\sim^{s} W''$.
\end{itemize}
\end{lemma}

\begin{proof}
Reflexivity and symmetry are straightforward. We only prove transitivity. Let $\text{CUT}^{1}$ be a common minimum cut separating $W$ from $D_{W}$ and $W'$ from $D_{W'}$, i.e., $\text{CUT}^{1}\in \operatorname{MINCUT}(D_{W},W)\cap \operatorname{MINCUT}(D_{W'},W')$. Similarly, let $\text{CUT}^{2}$ be a common minimum cut separating $W'$ from $D_{W'}$ and $W''$ from $D_{W''}$, i.e., $\text{CUT}^{2}\in \operatorname{MINCUT}(D_{W'},W')\cap \operatorname{MINCUT}(D_{W''},W'')$. Clearly,
\[|\text{CUT}^{1}|=\operatorname{mincut}(D_{W'},W')=|\text{CUT}^{2}|.
\]
Let \(m \triangleq \operatorname{mincut}(D_{W'},W')\). By Menger’s Theorem, there exist $m$ edge-disjoint paths $P_{1},P_{2},\dots,P_{m}$ from $D_{W'}$ to $W'$, and no additional paths exist after removing these paths. By the definition of minimum cuts, we have $\text{CUT}^{j}=\{e_{j,i} :1\leq i \leq m\}$ for $j = 1,2$, where $e_{j,i}$ lies on $P_{i}$.

We define a new edge set:
\[\text{CUT}=\{\operatorname{minord}(e_{1,i},e_{2,i}) : 1\leq i \leq m\},\]
where
\[\operatorname{minord}\left(e_{1, i}, e_{2, i}\right)=\left\{\begin{array}{ll}e_{1, i}, & e_{1, i} \prec e_{2, i}, \\ e_{2, i}, & \text { otherwise. }\end{array}\right.\]
We next prove $\text{CUT}$ is a common minimum cut separating $W$ from $D_{W}$ and $W''$ from $D_{W''}$. By symmetry, we only show $\text{CUT} \in \operatorname{MINCUT}(D_{W},W)$. Suppose, for contradiction, that there exists a path $P_{W}$ from $D_{W}$ to $W$ after removing all edges in $\text{CUT}$ from $\mathcal{G}$. Since $\text{CUT}^{1}$ is a minimum cut separating $W$ from $D_{W}$, $P_{W}$ must pass through $\text{CUT}^{1}$. Thus, there exists $i'\in[m]$ such that $P_{W}$ passes through $e_{1,i'}\in \text{CUT}^{1}$, and for any other edge $e_{1,i}\in \text{CUT}^{1}$ that $P_{W}$ passes through, $e_{1,i'} \prec e_{1,i}$. Since $e_{1,i'} \notin \text{CUT}$, we have $e_{2,i'}\prec e_{1,i'}$.

Let $P_{W}^{1}$ denote the subpath of $P_{W}$ from $D_{W}$ to $e_{1,i'}$. We claim $P_{W}^{1}$ does not pass through $\text{CUT}^{2}$. If it did, there would exist $i^{*}\in [m]$ such that $e_{1,i^{*}}\prec e_{2,i^{*}}$ and $P_{W}^{1}$ passes through $e_{2,i^{*}}$. We could then construct a path $P^{*}$ from $D_{W}$ to $W'$ by concatenating the subpath of $P_{W}^{1}$ from $D_{W}$ to $e_{2,i^{*}}$ with the subpath of $P_{i^{*}}$ from $e_{2,i^{*}}$ to $W'$. However, $P^{*}$ does not pass through $\text{CUT}^{1}$, contradicting the fact that $\text{CUT}^{1}$ is a minimum cut separating $W'$ from $D_{W'}$. Thus, $P_{W}^{1}$ does not pass through $\text{CUT}^{2}$.

Next, we construct a path $P^{W'}$ from some source node to $W'$ by concatenating $P_{W}^{1}$ with the subpath of $P_{i'}$ from $e_{1,i'}$ to $W'$. Neither subpath passes through $\text{CUT}^{2}$, so $P^{W'}$  does not pass through $\text{CUT}^{2}$, contradicting the fact that $\text{CUT}^{2}$ is a minimum cut separating $W'$ from $D_{W'}$.

Our initial assumption is false, so $\text{CUT} \in \operatorname{MINCUT}(D_{W},W)$. Similarly, we can prove $\text{CUT} \in \operatorname{MINCUT}(D_{W''},W'')$, establishing transitivity.
\end{proof}

We next consider the collection $\mathcal{W}_{r}^{*}$ for a non-negative integer $r$. We define another binary relation ``$\leq$'' on $\mathcal{W}_{r}^{*}$: for any $\alpha, \beta \in \mathcal{W}_{r}^{*}$, $\alpha \leq \beta$ if and only if $\alpha \in \operatorname{MINCUT}(D_{\beta},\beta)$. This relation is a partial order:

\begin{lemma}
The binary relation ``$\leq$'' on $\mathcal{W}_{r}^{*}$ is a partial order.
\end{lemma}

\begin{proof}
Reflexivity and transitivity are obvious. For antisymmetry, let $W, W' \in \mathcal{W}_{r}^{*}$ satisfy $W\leq W'$ and $W'\leq W$. By definition, $W$ separates $W'$ from $D_{W'}$ and $W'$ separates $W$ from $D_{W}$. Since $D_{W}=D_{W'}$, this implies $W = W'$.
\end{proof}

The equivalence relation ``$\sim^{s}$'' partitions $\mathcal{W}_{r}^{*}$ into distinct equivalence classes. We show that each class, partially ordered by ``$\leq$'', forms a meet-semilattice.

\begin{theorem}\label{thm:meet}
Let $Cl_{r}$ be an equivalence class of $\mathcal{W}_{r}^{*}$ under ``$\sim^{s}$''. Then $(Cl_{r}, \leq)$ is a meet-semilattice.
\end{theorem}

\begin{proof}
We show that for any $W, W' \in Cl_{r}$, their meet $W\wedge W'$ (infimum of $\{W, W'\}$) exists. Define the set of lower bounds:
\[\{W,W'\}^{l}=\{\eta\in Cl_{r}: \eta\leq W\ \text{and}\ \eta\leq W'\}.\]
Since there exists a common minimum cut separating $W$ from $D_{W}$ and $W'$ from $D_{W'}$, $\{W,W'\}^{l}$ is non-empty. Since $\{W,W'\}^{l}$ is a finite set, it suffices to show that for any $\eta_{1}, \eta_{2} \in \{W,W'\}^{l}$, there exists $\eta\in \{W,W'\}^{l}$ such that $\eta_{1}\leq \eta$ and $\eta_{2}\leq \eta$. This ensures the existence of a maximum element in $\{W,W'\}^{l}$, which is exactly the infimum $W\wedge W'$.

Let $m\triangleq \operatorname{mincut}(D_{W},W)$. By Menger’s Theorem, there exist $m$ edge-disjoint paths $P_{1}, P_{2}, \dots , P_{m}$ from $D_{W}$ to $W$, and no additional paths from $D_{W}$ to $W$ can exist after removing these $m$ paths. Since both $\eta_{1}$ and $\eta_{2}$ are minimum cuts separating $W$ from $D_{W}$, we have $\eta_{j}=\{e_{j,i} :1\leq i \leq m\}$ for $j = 1,2$, where $e_{j,i}$ lies on $P_{i}$. We construct a new edge set $\eta$ as
\[\eta=\{\operatorname{maxord}(e_{1,i},e_{2,i}) : 1\leq i \leq m\},\]
where
\[\operatorname{maxord}\left(e_{1, i}, e_{2, i}\right)=\left\{\begin{array}{ll}e_{2, i}, & e_{1, i} \prec e_{2, i}, \\ e_{1, i}, & \text { otherwise. }\end{array}\right.\]

We next prove that $\eta \in \{W,W'\}^{l}$ by verifying three key conditions: $\eta$ is a minimum cut separating $W$ from $D_{W}$, $\eta$ is a minimum cut separating $W'$ from $D_{W'}$, and $\eta\in Cl_{r}$. We first show that $\eta$ is a cut separating $W$ from $D_{W}$ by contradiction. Suppose, for contradiction, that there exists a path $P_{W}$ from $D_{W}$ to $W$ in $\mathcal{G}$ after removing all edges in $\eta$. Since $\eta_{1}$ and $\eta_{2}$ are both minimum cuts separating $W$ from $D_{W}$, $P_{W}$ must pass through at least one edge from each of $\eta_{1}$ and $\eta_{2}$. Let
\[\eta_{1}^{c}\triangleq \eta_{1}\backslash\eta \quad \text{and} \quad \eta_{2}^{c}\triangleq \eta_{2}\backslash\eta.\]
It follows that $P_{W}$ must pass through edges in both $\eta_{1}^{c}$ and $\eta_{2}^{c}$. Thus, there exist $i'\in[m]$ and $j'\in \{1,2\}$ such that $P_{W}$ passes through $e_{j',i'}\in \eta_{j'}$, and for any other edge $e \in \eta_{1}^{c}\cup \eta_{2}^{c}$ traversed by $P_{W}$, we have $e \prec e_{j',i'}$. We then construct a path $P^{*}$ from a source node to $W$ by concatenating the subpath of $P_{i'}$, which is a path from $D_{W}$ to $e_{j',i'}$, with the subpath of $P_{W}$, which is a path from $e_{j',i'}$ to $W$. Let $j''=3-j'$. It is easy to observe that $P^{*}$ does not pass through any edge in $\eta_{j''}$, which contradicts the fact that $\eta_{j''}$ is a minimum cut separating $W$ from $D_{W}$. Therefore, our assumption is false, and $\eta$ is indeed a cut separating $W$ from $D_{W}$. Since $|\eta|=m=\operatorname{mincut}(D_{W},W)$, $\eta$ is a minimum cut separating $W$ from $D_{W}$, which implies $\eta\leq W$.

By analogous reasoning, $\eta$ is also a minimum cut separating $W'$ from $D_{W'}$, so $\eta\leq W'$. Additionally, since $\operatorname{mincut}(D_{\eta},\eta)=|\eta|$, we have $\eta\in Cl_{r}$. Combining these results, we conclude that $\eta\in \{W,W'\}^{l}$.

We further prove that $\eta_{1}\leq \eta$. If $\eta_{1}=\eta$, the conclusion holds trivially. Assume $\eta_{1}\neq \eta$. Then $\eta\backslash \eta_{1}\neq\emptyset$. We claim that $\eta_{1}$ is a cut separating $\eta$ from $D_{\eta}$. Suppose not, then there exists a path $P_{\eta}$ from $D_{\eta}$ to some edge $e'\in \eta\backslash \eta_{1}$ in $\mathcal{G}$ after removing all edges in $\eta_{1}$. Since $\eta\backslash \eta_{1}\subseteq \eta_{2}$, we have $e'\in \eta_{2}$, i.e., $e'=e_{2,i}$ for some $i\in[m]$. We construct a path $P'$ from a source node to $W$ by concatenating the subpath of $P_{\eta}$, which is a path from $D_{\eta}$ to $e_{2,i}$, with the subpath of $P_{i}$, which is a path from $e_{2,i}$ to $W$. However, $P'$ does not pass through any edge in $\eta_{1}$, which contradicts the fact that $\eta_{1}$ is a minimum cut separating $W$ from $D_{W}$. Thus, $\eta_{1}$ is a cut separating $\eta$ from $D_{\eta}$. Furthermore, since $\operatorname{mincut}(D_{\eta},\eta)=|\eta|=|\eta_{1}|$, $\eta_{1}$ is a minimum cut separating $\eta$ from $D_{\eta}$, which implies $\eta_{1}\leq \eta$.

By symmetry, we also have $\eta_{2}\leq \eta$. In summary, the set $\{W,W'\}^{l}$ has a maximum element $\eta$, which is the infimum of $\{W, W'\}$. Therefore, the theorem holds.
\end{proof}

From \cref{prop:bottom-top-exist}, the meet-semilattice $(Cl_{r}, \leq)$ has a unique bottom element $\perp$. This leads to the following lemma:

\begin{lemma}\label{lem:exist bottom}
Let $Cl_{r}$ be an equivalence class of $\mathcal{W}_{r}^{*}$ under ``$\sim^{s}$''. Then there exists a unique element $\perp \in Cl_{r}$ such that $\perp$ is a minimum cut separating every other element of $Cl_{r}$ from all source nodes.
\end{lemma}

We now present a key theorem simplifying the upper bound:

\begin{theorem}
Let $Cl_{r}$ be an equivalence class of $\mathcal{W}_{r}^{*}$ under ``$\sim^{s}$'', and let $\perp$ be the bottom element of $(Cl_{r}, \leq)$. Then,
\[\Omega(\perp)\leq \Omega(W), \ \text{for all } W\in Cl_{r}.\]
\end{theorem}

\begin{proof}
It follows directly from Lemma \ref{lem:exist bottom} and Lemma \ref{lem:relation of omega}.
\end{proof}

Let $\widehat{\mathcal{W}}_{r}^{*}$ denote the set of bottom elements of all meet-semilattices $(Cl_{r},\leq)$ corresponding to the equivalence classes $Cl_{r}$ of $\mathcal{W}_{r}^{*}$ under the equivalence relation ``$\sim^{s}$''. Then, the upper bound in Theorem \ref{cor:the upper bound} can be rewritten as
\begin{align}\label{eq:bound form}
\hat{\mathcal{C}}(\mathcal{N},T,r)\leq \min_{W\in \widehat{\mathcal{W}}_{r}^{*}} \Omega(W).
\end{align}

Fix $W\in \widehat{\mathcal{W}}_{r}^{*}$. Next, we present a simplified form of the expression $\Omega(W)$ defined in \eqref{eq:Omega}. Let $\mathcal{G}_{W}$ denote the residual graph obtained by removing the edges in $W$ from $\mathcal{G}$, and let $\mathcal{N}_{W} \triangleq (\mathcal{G}_{W}, S, \rho)$. We show that $\Omega(W)$ can be computed within $\mathcal{N}_{W}$:

\begin{lemma}
Let $W\subseteq \mathcal{E}$ be an edge subset. Then
\begin{align}\label{eq:omega}
\Omega(W)=\min_{\substack{C\in \Lambda(\mathcal{N}_{W}):\\ D_{W}\subseteq I^{W}_{C}}}\frac{|C|}{\operatorname{Rank}(T_{I^{W}_{C}})},
\end{align}
where
\[I^{W}_{C}=\{\sigma\in S: \sigma \nrightarrow \rho \ \text{in}\ \mathcal{G}_{W}\ \text{after deleting the edges in}\ C\}.\]
\end{lemma}

\begin{proof}
Let $C \in \Lambda(\mathcal{N})$ be an edge subset belonging to
\[\operatorname{arg} \min_{C \in \Lambda(\mathcal{N})}\left\{\frac{|C\backslash W|}{\operatorname{Rank}(T_{I_{C}})}: W \subseteq C \ \text{and}\ D_{W} \subseteq I_{C}\right\}.\]
We claim $C\backslash W$ is a cut separating $\rho$ from $D_{W}$ in $\mathcal{G}_{W}$. If not, there exists a path $P$ from some $\sigma \in D_{W}$ to $\rho$ in $\mathcal{G}_{W}$ that does not pass through $C\backslash W$. This path would also be a path in $\mathcal{G}$ from $\sigma$ to $\rho$ that does not pass through $C$, which contradicts the fact that $C$ is a cut separating $\rho$ from $D_{W}$ in $\mathcal{G}$. Thus, \(C\backslash W\in  \Lambda(\mathcal{N}_{W})\) and \(D_{W}\subseteq I^{W}_{C\backslash W}\). Furthermore, \(I_{C}= I^{W}_{C\backslash W}\), so
\begin{align*}
\min_{\substack{C\in \Lambda(\mathcal{N}_{W}):\\ D_{W}\subseteq I^{W}_{C}}}\frac{|C|}{\operatorname{Rank}(T_{I^{W}_{C}})} \leq  \frac{|C\backslash W|}{\operatorname{Rank}(T_{I^{W}_{C\backslash W}})} =\frac{|C|-| W|}{\operatorname{Rank}(T_{I_{C}})} =\Omega(W).
\end{align*}

Conversely, let $C \in \Lambda(\mathcal{N}_{W})$ be an edge subset belonging to
\[\operatorname{arg} \min_{C \in \Lambda(\mathcal{N}_{W})}     \left\{\frac{|C|}{\operatorname{Rank}(T_{I^{W}_{C}})}: D_{W} \subseteq I_{C}^{W}\right\}.\]
We claim $C \cup W$ is a cut separating $\rho$ from $D_{W}$ in $\mathcal{G}$. If not, there exists a path from some $\sigma \in D_{W}$ to $\rho$ in $\mathcal{G}$ that does not pass through $C \cup W$, which implies a path in $\mathcal{G}_{W}$ from $\sigma$ to $\rho$ that does not pass through $C$, contradicting the choice of $C$. Thus, \(C \cup W \in \Lambda(\mathcal{N})\) and \(D_{W}\subseteq I_{C \cup W}\). Note that \(I^{W}_{C}=I_{C\cup W}\), so
\[
\Omega(W)\leq \frac{|C\cup W|-|W|}{\operatorname{Rank}(T_{I_{C\cup W}})} = \frac{|C|}{\operatorname{Rank}(T_{I_{C}^{W}})} =\min_{\substack{C\in \Lambda(\mathcal{N}_{W}):\\ D_{W}\subseteq I^{W}_{C}}}\frac{|C|}{\operatorname{Rank}(T_{I^{W}_{C}})}.
\]
This completes the proof.
\end{proof}

\begin{lemma}\label{lem:cut_rank_relation}
Let \(W \subseteq \mathcal{E}\) be an edge set. For any \(C\in \Lambda(\mathcal{N}_{W})\) with \(D_{W}\subseteq I^{W}_{C}\), if \(C'\) is a minimum cut separating \(\rho\) from \(C\) in \(\mathcal{G}_{W}\), then \(C'\in \Lambda(\mathcal{N}_{W})\), \(D_{W}\subseteq I^{W}_{C'}\), and
\[\frac{|C'|}{\operatorname{Rank}(T_{I^{W}_{C'}})}\leq \frac{|C|}{\operatorname{Rank}(T_{I^{W}_{C}})}.\]
\end{lemma}

\begin{proof}
Since $C'$ is a minimum cut separating $\rho$ from $C$, $|C'|\leq|C|$. We first prove $I^{W}_{C}\subseteq I^{W}_{C'}$ by contradiction. Suppose $I^{W}_{C}\nsubseteq I^{W}_{C'}$; then there exists a path $P$ from $\sigma\in I^{W}_{C}\backslash I^{W}_{C'}$ to $\rho$ in $\mathcal{G}_{W}$ that does not pass through $C'$. Since $C$ separates $\rho$ from $I^{W}_{C}$, $P$ must pass through an edge $e\in C$. This implies  the existence of a path from $e$ to $\rho$ in $\mathcal{G}_{W}$ that does not pass
through $C'$, contradicting the fact that $C'$ separates $\rho$ from $C$. Thus, $I^{W}_{C}\subseteq I^{W}_{C'}$, so \(\operatorname{Rank}(T_{I^{W}_{C}})\leq \operatorname{Rank}(T_{I^{W}_{C'}})\). Combining with $|C'|\leq|C|$, we obtain
\[\frac{|C'|}{\operatorname{Rank}(T_{I^{W}_{C'}})}\leq \frac{|C|}{\operatorname{Rank}(T_{I^{W}_{C}})}.\]
Since \(D_{W}\subseteq I^{W}_{C}\), we have \(D_{W}\subseteq I^{W}_{C'}\). This verifies that \(C'\) is a cut set in \(\mathcal{N}_{W}\) satisfying \(D_{W}\subseteq I^{W}_{C'}\). 
\end{proof}

By Lemma \ref{lem:cut_rank_relation}, we can restrict our attention to cuts \(C\) in $\mathcal{N}_{W}$ satisfying \(|C| = \text{mincut}^{W}(C,\rho)\), where $\text{mincut}^{W}(C,\rho)$ is the minimum cut capacity between $C$ and $\rho$ in $\mathcal{N}_W$. This is because Lemma \ref{lem:cut_rank_relation} guarantees that no other cuts can yield a smaller value of $\frac{|C|}{\operatorname{Rank}(T_{I^{W}_{C}})}$. Define the set $\mathcal{C}_{W}^{*}$ as
\[\mathcal{C}_{W}^{*}=\{C\in \Lambda(\mathcal{N}_{W}): |C|=\text{mincut}^{W}(C,\rho)\ \text{and}\ D_{W}\subseteq I_{C}^{W}\}.\]
Then, $\Omega(W)$ can be rewritten as
\begin{align}
\Omega(W)=\min_{C\in \mathcal{C}_{W}^{*}}\frac{|C|}{\operatorname{Rank}(T_{I^{W}_{C}})}.
\end{align}

We next introduce binary relations to simplify $\Omega(W)$ further. For two edge subsets $\alpha, \beta $ in $\mathcal{G}_{W}$, define ``$\sim^{\rho}$'': $\alpha \sim^{\rho} \beta$ if and only if there exists a common minimum cut separating $\rho$ from both $\alpha$ and $\beta$. As proven in \cite{GY2018}, ``$\sim^{\rho}$'' is an equivalence relation, which partitions the edge subsets in $\mathcal{C}_W^*$ into distinct equivalence classes. We next define a binary relation ``$\leq^{\rho}$'' on $\mathcal{C}_{W}^*$: for any $\alpha, \beta \in \mathcal{C}_{W}^*$, $\alpha \leq^{\rho} \beta$ if $\beta$ is a minimum cut separating $\rho$ from $\alpha$. The following lemma is cited from \cite{GCY2025}:

\begin{lemma}[{\cite[Theorem 8]{GCY2025}}]
The binary relation ``$\leq^{\rho}$'' on $\mathcal{C}_{W}^*$ is a partial order. Let $Cl$ be an equivalence class of $\mathcal{C}_{W}^*$ under ``$\sim^{\rho}$''. Then $(Cl,\leq^{\rho})$ is a join-semilattice.
\end{lemma}

From Proposition \ref{prop:bottom-top-exist}, each join-semilattice $(Cl, \leq^{\rho})$ has a unique top element $\top$, which is a minimum cut separating $\rho$ from every $C \in Cl$. By Lemma \ref{lem:cut_rank_relation}, we have
\[\frac{|\top|}{\operatorname{Rank}(T_{I_{\top}^{W}})}\leq \frac{|C|}{\operatorname{Rank}(T_{I_{C}^{W}})} \ \text{for all } C \in Cl.\]
Thus, we only need to consider these top elements when computing $\Omega(W)$:

\begin{theorem}\label{thm:omega form}
For a wiretap set $W\in \mathcal{W}_{r}$, let $\widehat{\mathcal{C}}_{W}^{*}$ denote the set of top elements of all join-semilattices $(Cl,\leq^{\rho})$ corresponding to the equivalence classes $Cl$ of $\mathcal{C}_{W}^*$ under ``$\sim^{\rho}$''. Then
\[\Omega(W)=\min_{C\in \widehat{\mathcal{C}}_{W}^{*}} \frac{|C|}{\operatorname{Rank}(T_{I^{W}_{C}})}.\]
\end{theorem}

Combining Theorem \ref{thm:omega form} and \eqref{eq:bound form}, we obtain the simplest form of the upper bound:

\begin{theorem}\label{cor:bound}
The upper bound on the secure computing capacity $\hat{\mathcal{C}}(\mathcal{N},T,r)$ is
\[\hat{\mathcal{C}}(\mathcal{N},T,r)\leq \min_{W\in \widehat{\mathcal{W}}_{r}^{*}}\min_{C\in\widehat{\mathcal{C}}_{W}^{*}} \frac{|C|}{\operatorname{Rank}(T_{I^{W}_{C}})}.\]
\end{theorem}

Note that when $k=1$, the target function $f$ reduces to a scalar linear function over $\mathbb{F}_q$. In \cite{GBY2024}, Guang et al. proposed an upper bound for this setting and simplified it using primary minimum cuts. We recall their results below.

The primary minimum cut of an edge subset $W$ is a minimum cut that not only separates $W$ from $D_{W}$ but also separates all other minimum cuts separating $W$ from $D_{W}$ from $D_{W}$. Let $\hat{W}$ denote the primary minimum cut of $W$, then $W$ is primary if $\hat{W}=W$. Let $\hat{\mathcal{W}}_r = \{W\in \mathcal{W}_{r}: \hat{W}=W\}$ denote the set of primary wiretap sets. Furthermore, let $C_{W}$ denote the primary minimum cut separating $\rho$ from $D_{W}$ in $\mathcal{G}_W$, i.e., $C_{W}$ is a minimum cut separating $\rho$ from $D_{W}$ that also separates $\rho$ from all other minimum cuts separating $\rho$ from $D_{W}$ in $\mathcal{G}_W$. As shown in \cite[Theorem 9]{GBY2024}, the upper bound simplifies to:
\begin{align}\label{eq:exist bound}
\hat{\mathcal{C}}(\mathcal{N},f,r)\leq \min_{W\in \hat{\mathcal{W}}_{r}} |C_{W}|,
\end{align}
where $f$ is the algebraic sum over $\mathbb{F}_q$.

Since scalar linear functions are a special case of linear functions, we now verify that the bound in \eqref{eq:exist bound} can be directly derived from our upper bound in Theorem \ref{cor:bound}. First, Guang et al. investigated the relationship between primary minimum cuts and the bottom (top) elements of meet-semilattices (join-semilattices) in \cite{GCY2025}. Through a proof analogous to {\cite[Theorem 13]{GCY2025}}, we can show that if $W\in \mathcal{W}_{r}$ is primary, then $W\in \widehat{\mathcal{W}}_{r}^{*}$. On the other hand, for any bottom element $\perp$ in $\mathcal{W}_r^{*}$, we claim that $\perp$ is primary. Otherwise, there exists a minimum cut $\perp'$ separating $\perp$ from $D_{\perp}$ such that $\perp$ cannot separate $\perp'$ from $D_{\perp'}$. Since $\perp'$ separates $\perp$ from $D_{\perp}$, it follows that $\perp'$ lies in the same equivalence class as $\perp$. Moreover, since $\perp$ is a bottom element, $\perp$ separates $\perp'$ from $D_{\perp'}$, which yields a contradiction. This shows that for a wiretap set $W\in \mathcal{W}_{r}$, $W$ is primary if and only if $W$ is the bottom element of the meet-semilattice $(Cl_{r}, \leq)$, where $Cl_{r}$ is some equivalence class of $\mathcal{W}_{r}^{*}$ under ``$\sim^{s}$''. This implies that $\widehat{\mathcal{W}}_r^{*} = \hat{\mathcal{W}}_r$.

For the algebraic sum over \(\mathbb{F}_{q}\), \(T\) is a column vector with all nonzero entries, so $\operatorname{Rank}(T_{I^{W}_{C}})=1$ for any non-empty $I^{W}_{C}$. Additionally, all minimum cuts separating $\rho$ from $D_W$ in $\mathcal{G}_W$ form a single equivalence class in $\mathcal{C}_W^*$ under ``$\sim^{\rho}$''. Let $\top_{min}$ be the top element of this equivalence class under the partial order $\leq^{\rho}$, then $\top_{min}\in \widehat{\mathcal{C}}_{W}^{*}$. Furthermore, we can prove that $\top_{min}=C_{W}$. Since the size of any cut separating $\rho$ from $D_W$ in $\mathcal{G}_W$ is no less than the size of the minimum cut separating $\rho$ from $D_W$, we have
\begin{align*}
\hat{\mathcal{C}}(\mathcal{N},f,r) & \leq  \min_{W\in \widehat{\mathcal{W}}_{r}^{*}}\min_{C\in\widehat{\mathcal{C}}_{W}^{*}}  \frac{|C|}{\operatorname{Rank}(T_{I^{W}_{C}})} \\
&= \min_{W\in \hat{\mathcal{W}}_{r}}\min_{C\in \widehat{\mathcal{C}}_{W}^{*}} |C|\\
&= \min_{W\in \hat{\mathcal{W}}_{r}} |\top_{min}|\\
&=\min_{W\in \hat{\mathcal{W}}_{r}} |C_{W}|.
\end{align*}

This shows that our upper bound reduces exactly to \cite[Theorem 9]{GBY2024} for scalar linear functions, demonstrating that Theorem \ref{cor:bound} generalizes existing results to vector linear target functions and provides a more universal tool for computing the upper bound on the secure computing capacity.

Theorem \ref{cor:bound} provides an efficient theoretical form for the upper bound on the secure computing capacity, but it is not directly computable. We thus develop an efficient algorithm to compute this bound. First, for any wiretap set $W\in \mathcal{W}_{r}$, we develop an efficient algorithm to find the collection $\widehat{\mathcal{C}}_{W}^{*}$, which consists of all primary cut sets separating $\rho$ from $D_{W}$ in $\mathcal{G}_{W}$. An implementation of this algorithm is given in Algorithm \ref{alg:algorithm1}. In Step 5, we use the path augmentation algorithm to find $\operatorname{mincut}(N, \rho)$ edge-disjoint paths in $\mathcal{O}(\operatorname{mincut}(N, \rho)\cdot |\mathcal{E}|)$ time. In Step 6, we modify the linear-time algorithm from \cite[Algorithm 4]{GCY2025} to find the primary minimum cut separating a node from a subset of nodes. In Step 8, we modify the algorithm from \cite[Algorithm 2]{GCY2025} to find $V_{\top}^{c}$ in $\mathcal{O}(|V_{\top}^{c}|)$ time. Combining these steps, we can obtain the collection $\widehat{\mathcal{C}}_{W}^{*}$ using Algorithm \ref{alg:algorithm1}, whose time complexity is linear in the number of edges in the graph. 

Finally, we develop an efficient algorithm to compute the upper bound in Theorem \ref{cor:bound}, which is presented in Algorithm \ref{alg:algorithm2}. In Step 3, we adapt \cite[Algorithm 1]{GCY2025} to find $\widehat{\mathcal{W}}_{r}^{*}$ in $\mathcal{O}(|\widehat{\mathcal{W}}_{r}^{*}|\cdot|\mathcal{E}|)$ time. In Step 5, we use Algorithm \ref{alg:algorithm1} to find $\widehat{\mathcal{C}}_{W}^{*}$. Thus, Algorithm \ref{alg:algorithm2} runs in linear time in the number of edges in the graph.

\begin{algorithm}[htbp]
\caption{Algorithm for Computing $\widehat{\mathcal{C}}_{W}^{*}$}\label{alg:algorithm1}
  \begin{algorithmic}[1]
\REQUIRE The network $\mathcal{N}_{W}=(\mathcal{G}_{W},D_{W},\rho)$, where $\mathcal{G}_{W}=(V,E)$.
\ENSURE $\widehat{\mathcal{C}}_{W}^{*}$, the set of all primary global cut sets in $\Lambda(\mathcal{N}_{W})$ for $\rho$.
\STATE Set $\widehat{\mathcal{C}}_{W}^{*} = \emptyset$;
\STATE Set $\mathscr{A} = \{N \subseteq V : D_{W}\subseteq N\subseteq (V_{\rho}\backslash \{\rho\})\}$, where $V_{\rho}$ is the set of nodes from which $\rho$ is reachable; \COMMENT{$\rho$ is reachable from $v$ if there exists a directed path from $v$ to $\rho$}
\WHILE{$\mathscr{A}\neq\emptyset$}
\STATE Choose a subset $N \in \mathscr{A}$;
\STATE Find $\operatorname{mincut}(N, \rho)$ edge-disjoint paths from $N$ to $\rho$;
\STATE Find the primary minimum cut $\top$ separating $\rho$ from $N$;
\STATE Add $\top$ to $\widehat{\mathcal{C}}_{W}^{*}$;
\STATE Partition $V_{\rho}\backslash\{\rho\}$ into $V_{\top}$ (the set of the nodes from which $\rho$ is reachable after deleting the edges in $\top$) and $V_{\top}^{c}=V_{\rho}\backslash(\{\rho\}\cup V_{\top})$;
\FOR{each $\mu\in\mathscr{A}$}
\IF{$\mu\subseteq V_{\top}^{c}$}
\STATE Remove $\mu$ from $\mathscr{A}$;
\ENDIF
\ENDFOR
\ENDWHILE
\RETURN $\widehat{\mathcal{C}}_{W}^{*}$.
  \end{algorithmic}
\end{algorithm}

\begin{algorithm}[htbp]
\caption{Algorithm for Computing \(\min_{W\in \widehat{\mathcal{W}}_{r}^{*}}\min_{C\in\widehat{\mathcal{C}}_{W}^{*}} \frac{|C|}{\operatorname{Rank}(T_{I^{W}_{C}})}\)}\label{alg:algorithm2}
  \begin{algorithmic}[1]
\REQUIRE The network $\mathcal{N}=(\mathcal{G},S,\rho)$ ($\mathcal{G}=(V,E)$, $S=\{\sigma_{1},\sigma_{2},\dots,\sigma_{s}\}$), a non-negative integer $r$, and a coefficient matrix $T$.
\ENSURE $\min_{W\in \widehat{\mathcal{W}}_{r}^{*}}\min_{C\in\widehat{\mathcal{C}}_{W}^{*}} \frac{|C|}{\operatorname{Rank}(T_{I^{W}_{C}})}$.
\STATE Set $U=\infty$;
\STATE Set $\Omega=\infty$;
\STATE Find $\widehat{\mathcal{W}}_{r}^{*}$, the set of all primary wiretap sets in $\mathcal{W}_{r}$;
\FOR{each $W\in \widehat{\mathcal{W}}_{r}^{*}$}
\STATE Find $\widehat{\mathcal{C}}_{W}^{*}$ in the network $\mathcal{N}_{W}=(\mathcal{G}_{W},D_{W},\rho)$;
\FOR{each $C\in \widehat{\mathcal{C}}_{W}^{*}$}
\STATE Find $I_{C}^{W}$, the set of source nodes from which $\rho$ is unreachable after deleting the edges in $C$ in $\mathcal{G}_{W}$;
\STATE Compute $\frac{|C|}{\operatorname{Rank}(T_{I_{C}^{W}})}$;
\STATE $\Omega=\min(\Omega,\frac{|C|}{\operatorname{Rank}(T_{I_{C}^{W}})})$;
\ENDFOR
\STATE $U=\min(U,\Omega)$;
\ENDFOR
\RETURN $U$.
  \end{algorithmic}
\end{algorithm}

\section{Linear Secure Network Codes for Linear Target Functions under Linear Security Functions}\label{sec:linear_codes}

In this section, we study linear secure network coding for linear target functions and linear security functions. We first formalize the code model, then propose two constructions of admissible linear secure network codes for the
secure model $(N, f, \mathcal{W}_{r},\zeta)$, which yield a lower bound on the secure computing capacity.
\subsection{Linear Function-Computing Secure Network Coding}

Let \( f, \zeta \) be linear functions over a finite field \( \mathbb{F}_q \), that is, \( f(\vec{m}_{S}) = \vec{m}_{S} \cdot T \) and \( \zeta(\vec{m}_{S}) = \vec{m}_{S} \cdot \Upsilon \), where \( T \in \mathbb{F}_q^{s \times k} \), \( \Upsilon \in \mathbb{F}_q^{s \times r_{\Upsilon}} \). Without loss of generality, assume that both $T$ and $\Upsilon$ are of full column rank. A secure network code is called \textit{linear} if the local encoding function associated with each edge is linear. Let $\ell$ and $n$ be two positive integers. We now define an $(\ell, n)$ linear secure network code over the finite field $\mathbb{F}_{q}$ for the model $(N, f, \mathcal{W}_{r},\zeta)$, following the framework in \cite{GBY2024}, \cite{arxiv2025secure}. Each source node $\sigma_{i}$ generates
$\mathbf{M}_{i}=(M_{i, 1}, M_{i, 2}, \cdots, M_{i, \ell})$, which is uniformly distributed over $\mathbb{F}_q^{\ell}$.
We denote $\mathbf{M}_{S}=(\mathbf{M}_{1}, \mathbf{M}_{2}, \cdots, \mathbf{M}_{s})$.
Additionally, $\mathbf{K}_{i}$ represents the random key available at $\sigma_{i}$, which is uniformly distributed over $\mathbb{F}_{q}^{z_{i}}$ (where $z_{i}$ is a nonnegative integer), and $\mathbf{K}_{S}=(\mathbf{K}_{1}, \mathbf{K}_{2}, \cdots, \mathbf{K}_{s})$. Let $\mathbf{m}_{i} \in \mathbb{F}_{q}^{\ell}$ and $\mathbf{k}_{i} \in \mathbb{F}_{q}^{z_{i}}$ be arbitrary realizations of $\mathbf{M}_{i}$ and $\mathbf{K}_{i}$, respectively. Correspondingly, $\mathbf{m}_{S}=(\mathbf{m}_{1}, \mathbf{m}_{2}, \cdots, \mathbf{m}_{s})$ and $\mathbf{k}_{S}=(\mathbf{k}_{1}, \mathbf{k}_{2}, \cdots, \mathbf{k}_{s})$ denote arbitrary realizations of $\mathbf{M}_{S}$ and $\mathbf{K}_{S}$. We further define $\mathbf{x}_{i}=(\mathbf{m}_{i} \ \mathbf{k}_{i})$ and $\mathbf{x}_{S}=(\mathbf{x}_{1} \ \mathbf{x}_{2} \ \cdots \ \mathbf{x}_{s})$.

An $(\ell, n)$ linear secure network code $\hat{C}=\{\hat{\theta}_{e}: e \in \mathcal{E}; \hat{\varphi}\}$ comprises a linear local encoding function $\hat{\theta}_{e}$ for each edge $e$ and a decoding function $\hat{\varphi}$ at the sink node $\rho$, defined as follows:
\begin{align}\label{eq:linear coding}
\left\{\begin{array}{ll} \hat{\theta}_{e}\left(\mathbf{x}_{i}\right)=\mathbf{x}_{i} \cdot A_{i, e}, & \text{if } \operatorname{tail}(e)=\sigma_{i} \text{ for some } i, \\ \hat{\theta}_{e}\left(y_{d}: d \in \operatorname{In}(v)\right)=\sum_{d \in \operatorname{In}(v)} y_{d} \cdot A_{d, e}, & \text{otherwise;} \end{array}\right.
\end{align}
where $v \triangleq \operatorname{tail}(e)$, $A_{i, e}\in \mathbb{F}_q^{(\ell+z_{i}) \times n}$, $A_{d, e}\in \mathbb{F}_q^{n\times n}$.
Moreover, the decoding function $\hat{\varphi}$ at $\rho$ maps $\prod_{\operatorname{In}(\rho)} \mathbb{F}_{q}^{n}$ to $\mathbb{F}_{q}^{\ell k}$.

Observe that every global encoding function $\hat{g}_{e}$, which is induced by the local encoding functions $\{\hat{\theta}_{e}: e \in \mathcal{E}\}$, is linear. Thus, there exists a matrix $\hat{\mathbf{g}}_{e}\in \mathbb{F}_q^{(\ell s+\sum_{i=1}^{s} z_{i}) \times n}$ such that
\[y_{e}=\hat{g}_{e}\left(\mathbf{x}_{S}\right)=\mathbf{x}_{S} \cdot \hat{\mathbf{g}}_{e},\]
where $\hat{\mathbf{g}}_{e}$ is referred to as the global encoding matrix for edge $e$.

For such an $(\ell, n)$ linear secure network code $\hat{C}$ for the model $(\mathcal{N}, f, \mathcal{W}_{r}, \zeta)$, define a matrix $\mathbf{S}(\hat{C})$ as follows:
\[
\mathbf{S}(\hat{C}) \triangleq 
\begin{bmatrix}
\mathbf{S}^{(\sigma_1)} \\
\mathbf{S}^{(\sigma_2)} \\
\vdots \\
\mathbf{S}^{(\sigma_s)}
\end{bmatrix},
\quad \text{with }
\mathbf{S}^{(\sigma_i)} = 
\begin{bmatrix}
\mathbf{\Upsilon}_i \otimes I_\ell \\
\mathbf{0}_{z_i \times \ell r_{\Upsilon}}
\end{bmatrix}, \ \forall i \in [s],
\]
where $I_\ell$ is the $\ell \times \ell$ identity matrix, $\Upsilon_i$ is the $i$-th row of $\Upsilon$, and $\otimes$ is the Kronecker product. In the rest of this paper, we write $\mathbf{S}(\hat{C})$ as $\mathbf{S}$ for notational simplicity when there is no ambiguity regarding the code $\hat{C}$. Following the same line of argument as in \cite[Theorem 7]{xu2026network}, we derive the following lemma, which presents an equivalent characterization of the security condition.

\begin{lemma}\label{lem:security wen}
Consider the secure model $(\mathcal{N}, f, \mathcal{W}_{r}, \zeta)$ with $f(\vec{m}_{S})=\vec{m}_{S}\cdot T$ and $\zeta(\vec{m}_{S})=\vec{m}_{S}\cdot \Upsilon$. Let $\hat{C}$ be an $(\ell, n)$ linear secure network code with the global encoding matrices $\{\hat{\mathbf{g}}_e, e \in \mathcal{E}\}$. Then, the security condition \eqref{eq:security condition} is satisfied for the code $\hat{C}$ if and only if
\[
\langle G_W \rangle \cap \langle \mathbf{S} \rangle = \{ \mathbf{0} \}, \quad \forall W \in \mathcal{W}_r,
\]
where $G_{W}\triangleq [\hat{\mathbf{g}}_e: e \in W]$.
\end{lemma}

\subsection{Code Construction}

In \cite{GBY2024}, \cite{arxiv2025secure}, and \cite{xu2026network}, a linear secure network code was constructed by transforming a non-secure code into a secure one. We extend this idea to the case where the target function $f$ is a vector linear function, rather than the algebraic sum.

Denote
\[
C_{\min} = \min \left\{ |C| : C \in \Lambda(\mathcal{N}) \right\}.
\]
Consider the network function computation model $(\mathcal{N}, f)$ without any security consideration. For any nonnegative integer $R$ satisfying $R \leq C_{\text{min}}$, we can construct an admissible $(R,k)$ linear network code $\mathcal{C}$ with the global encoding matrices $\{\mathbf{g}_{e},e\in\mathcal{E}\}$ for $(\mathcal{N},f)$. First, 
let $\mathcal{G}^{T}$ denote the reversed graph of $\mathcal{G}$, obtained by reversing the direction of every edge in $\mathcal{G}$ and treating $\rho$ as the single source node and $S$ as the set of sink nodes. Denote the reversed network as $\mathcal{N}^{T} \triangleq (\mathcal{G}^{T}, \rho, S)$. By reversing a rate-$R$ linear network code for the single-source multicast problem
on $\mathcal{N}^{T}$, we obtain an admissible $(R, 1)$ linear network code $C$ for the model $(\mathcal{N}, f)$ where $f$ is the algebraic sum. Let $T_{i}$ denote the $i$-th column of $T$. Through a linear transformation of the code $C$, we derive admissible $(R, 1)$ linear network codes $C_{i}$ for the model $(\mathcal{N}, f_{i})$ where $f_{i}$ is a scalar linear function with the efficient vector $T_{i}$. Combining these $k$ codes $C_{i}$ to form $\mathcal{C}$, which utilizes the network $k$ times. Then, $\mathcal{C}$ is an admissible $(R,k)$ linear network code for the model $(\mathcal{N}, f)$, where $f$ is a vector linear function with the coefficient matrix $T$.

We further decompose $\mathbf{g}_{e}$ as
\begin{align}\label{eq:glabal single}\mathbf{g}_{e}=\left[\begin{array}{c}\mathbf{g}_{e}^{(\sigma_{1})} \\ \mathbf{g}_{e}^{(\sigma_{2})} \\ \vdots \\ \mathbf{g}_{e}^{(\sigma_{s})}\end{array}\right],\end{align}
where each $\mathbf{g}_{e}^{(\sigma_{i})}$ is an $\mathbb{F}_{q}$-valued $R\times k$ matrix. For an edge subset $W \subseteq \mathcal{E}$, we define $G_{W} \triangleq[\mathbf{g}_{e}: e \in W]$ and decompose it similarly:
\[
G_{W}=\left[\begin{array}{c}G_{W}^{(\sigma_{1})} \\ G_{W}^{(\sigma_{2})} \\ \vdots \\ G_{W}^{(\sigma_{s})}\end{array}\right]
\]
where each $G_{W}^{(\sigma_{i})} \triangleq[\mathbf{g}_{e}^{(\sigma_{i})}: e \in W]$ is an $\mathbb{F}_{q}$-valued $R\times k |W|$ matrix.

Let $B_{i}$ be an $R \times R$ invertible matrix over $\mathbb{F}_{q}$ for $i=1,2,\cdots,s$, referred to as the transformation matrix, and define $\widehat{B}$ as an $s\times s$ block matrix with $B_{i}$ on the diagonal and $R \times R$ zero matrices elsewhere, i.e.,
\begin{align}\label{eq:B}\widehat{B}\triangleq \left[ B_1 \quad B_2 \quad \cdots \quad B_s \right]_{s}^{\text{diagonal}}=\left[\begin{array}{cccc}B_1 & 0 & \cdots & 0 \\ 0 & B_2 & \cdots & 0 \\ \cdots & \cdots & \cdots & \cdots \\ 0 & 0 & \cdots & B_s\end{array}\right]_{Rs \times Rs}.\end{align}
We denote the transformed code as $\widehat{\mathcal{C}} =\widehat{B} \cdot \mathcal{C}$, whose global encoding matrix $\widehat{\mathbf{g}}_{e}$ for each edge $e$ is $\widehat{B} \cdot \mathbf{g}_{e}$, i.e.,
\[
\widehat{\mathbf{g}}_{e}=\left[\begin{array}{c}\widehat{\mathbf{g}}_{e}^{(\sigma_{1})} \\ \widehat{\mathbf{g}}_{e}^{(\sigma_{2})} \\ \vdots \\ \widehat{\mathbf{g}}_{e}^{(\sigma_{s})}\end{array}\right]=\left[\begin{array}{c}B_1 \cdot \mathbf{g}_{e}^{(\sigma_{1})} \\ B_2 \cdot \mathbf{g}_{e}^{(\sigma_{2})} \\ \vdots \\ B_s \cdot \mathbf{g}_{e}^{(\sigma_{s})}\end{array}\right].
\]

Let $R_1$ be an arbitrary positive integer satisfying $R_1 \leq R$. Depending on the matrix $T$, we define a matrix $T_{R_{1}}$ as
\[
T_{R_{1}}=\left[\begin{array}{cccc} A^{11} & A^{12} & \cdots & A^{1k} \\ A^{21} & A^{22} & \cdots & A^{2k} \\ \vdots & \vdots & \vdots & \vdots \\ A^{s1} & A^{s2} & \cdots & A^{sk} \end{array}\right]
\]
where $A^{ij}=\left[\begin{array}{c} T_{i,j}\cdot I_{R_1} \\ \mathbf{0}_{(R-R_{1}) \times R_{1}}\end{array}\right]$ for all $1 \leq i \leq s$ and $1 \leq j \leq k$, and $T_{i,j}$ denotes the element in the $i$-th row and $j$-th column of $T$. 

\begin{theorem}\label{thm:condition}
Consider the secure model $(\mathcal{N}, f, \mathcal{W}_{r},\zeta)$, where $f(\vec{m}_{S})=\vec{m}_{S}\cdot T$ and $\zeta(\vec{m}_{S})=\vec{m}_{S}\cdot \Upsilon$. Let $R$ be any non-negative integer satisfying $rk \leq R \leq C_{\text{min}}$. The code $\mathcal{C}$ is constructed as above, and the matrix $\widehat{B}$ is defined as above. If $\widehat{B}$ satisfies the following two conditions:
\begin{itemize}
    \item  $\langle \widehat{B}^{-1} \cdot T_{R-rk} \rangle \subseteq \langle G_{\text{In}(\rho)} \rangle $,
    \item $\langle \widehat{B}^{-1} \cdot \mathbf{S} \rangle \cap \langle G_W \rangle = \{ \vec{0} \} \quad \text{with} \quad \mathbf{S}^{(\sigma_i)} = 
    \begin{bmatrix}
    \mathbf{\Upsilon}_i \otimes I_{R-rk} \\
    \mathbf{0}_{rk \times (R-rk) r_{\Upsilon}}
    \end{bmatrix}, \quad \forall\ W \in \mathcal{W}_r$,
\end{itemize}
the code $\widehat{\mathcal{C}}=\widehat{B} \cdot \mathcal{C}$ is an admissible $(R-rk, k)$ linear secure network code for the model $(\mathcal{N}, f, \mathcal{W}_{r},\zeta)$.
\end{theorem}

\begin{proof}
We first verify the computability condition in \eqref{eq:computability}. For $\widehat{\mathcal{C}}$ to satisfy the computability condition, $\rho$ must be able to compute 
\[f(\mathbf{x}_{S})=\mathbf{x}_{S}\cdot T_{R-rk}.\]
Using the code $\widehat{\mathcal{C}}$, the vector of messages received at $\rho$ is
    \begin{align}\label{eq:sink node}
    y_{\operatorname{In}(\rho)}=\left(\mathbf{x}_{S} \cdot \widehat{\mathbf{g}}_{e}: e \in \operatorname{In}(\rho)\right)=\left(\mathbf{x}_{S} \cdot\widehat{B} \mathbf{g}_{e}: e \in \operatorname{In}(\rho)\right)=\mathbf{x}_{S}\widehat{B}\cdot G_{\operatorname{In}(\rho) }.
    \end{align}
Since $\langle \widehat{B}^{-1} \cdot T_{R-rk} \rangle \subseteq \langle G_{\text{In}(\rho)} \rangle $, there exists a matrix $D$ such that $\widehat{B}^{-1} \cdot T_{R-rk}  =G_{\text{In}(\rho)} \cdot D$. Thus,
\[
    y_{\operatorname{In}(\rho)} \cdot D = \mathbf{x}_{S}\widehat{B}\cdot G_{\operatorname{In}(\rho)}\cdot D = \mathbf{x}_S \cdot T_{R-rk}.
\]
Hence, we have verified the computability condition of the code $\widehat{\mathcal{C}}$.

Since $\widehat{B}$ is invertible, $\langle \widehat{B}^{-1} \cdot \mathbf{S} \rangle \cap \langle G_W \rangle = \{ \vec{0} \}$ is equivalent to $\langle  \mathbf{S} \rangle \cap \langle \widehat{B} \cdot G_W \rangle = \{ \vec{0} \}$. From Lemma \ref{lem:security wen}, it follows that the code $\mathcal{C}$ satisfies the security condition \eqref{eq:security condition}.
\end{proof}

Note that the conditions in Theorem \ref{thm:condition} are abstract, so we provide a specific method to construct $R \times R$ invertible matrices $B_{1},B_{2},\cdots, B_{s}$ that satisfy the conditions in Theorem \ref{thm:condition}. We sequentially choose $R$ linearly independent column vectors $\vec{b}_{1}, \vec{b}_{2}, \cdots, \vec{b}_{R}\in\mathbb{F}_{q}^{R}$ as follows:
\begin{itemize}
    \item  For $1 \leq j \leq R-rk$, choose
     \begin{align}\label{eq:choice1}
        \vec{b}_{j} \in \mathbb{F}_{q}^{R} \backslash \bigcup_{W \in \mathcal{W}_{r}} \bigcup_{i=1}^{s}\left(\left\langle G_{W}^{(\sigma_{i})}\right\rangle+\left\langle \vec{b}_{1}, \vec{b}_{2}, \cdots, \vec{b}_{j-1}\right\rangle\right),
        \end{align}
    \item For $R-rk+1 \leq j \leq R$, choose
    \begin{align}\label{eq:choice2}
            \vec{b}_{j} \in \mathbb{F}_{q}^{R} \backslash \left\langle \vec{b}_{1}, \vec{b}_{2}, \cdots, \vec{b}_{j-1}\right\rangle.
        \end{align}
\end{itemize}
Let $B=[\begin{array}{llll}\vec{b}_{1} & \vec{b}_{2} & \cdots & \vec{b}_{R}\end{array}]^{-1}$ and $\widehat{B}=[\begin{array}{llll}B & B & \cdots & B\end{array}]_{s}^{\text{diagonal }}$. Under the above choice, we have 
\[
\widehat{B}^{-1}\mathbf{S}=\begin{bmatrix}
B^{-1}\mathbf{S}^{(\sigma_1)} \\
B^{-1}\mathbf{S}^{(\sigma_2)} \\
\vdots \\
B^{-1}\mathbf{S}^{(\sigma_s)}
\end{bmatrix}=\Upsilon \otimes [\begin{array}{cccc} \vec{b}_{1}& \vec{b}_{2}& \cdots& \vec{b}_{R-rk}\end{array}].
\]
Equation \eqref{eq:choice1} not only ensures that
\[
    \langle \widehat{B}^{-1} \cdot \mathbf{S} \rangle \cap \langle G_W \rangle = \{ \vec{0} \}  \quad \forall\ W \in \mathcal{W}_r,
\]
but also guarantees the linear independence of the vectors $\vec{b}_{1}, \vec{b}_{2}, \cdots, \vec{b}_{R-rk}$. Equation \eqref{eq:choice2} ensures the invertibility of the matrix $B$. Since the linear network code $\mathcal{C}$ is admissible, there exists a matrix $D\in \mathbb{F}_{q}^{|In(\rho)|k\times Rk}$ such that
\[
    G_{In(\rho)}\cdot D=T\otimes I_{R}.
\]
This implies that $\langle T\otimes I_{R} \rangle\subseteq \langle G_{In(\rho)} \rangle$. Observe that
 \begin{align*}
    \widehat{B}^{-1}T_{R-rk} =\left[\begin{array}{cccc}B^{-1}A^{11} & B^{-1}A^{12} & \cdots & B^{-1}A^{1k} \\ B^{-1}A^{21} & B^{-1}A^{22} & \cdots & B^{-1}A^{2k} \\ \vdots & \vdots & \vdots & \vdots \\ B^{-1}A^{s1} & B^{-1}A^{s2} & \cdots & B^{-1}A^{sk}\end{array}\right]= T \otimes  [\begin{array}{cccc} \vec{b}_{1}& \vec{b}_{2}& \cdots& \vec{b}_{R-rk}\end{array}].
    \end{align*}
This implies that each column of $\widehat{B}^{-1}T_{R-rk}$ can be linearly expressed using the columns of $T\otimes I_{R}$. Thus, $\langle \widehat{B}^{-1}T_{R-rk}\rangle\subseteq \langle T\otimes I_{R}\rangle$. Together with $\langle T\otimes I_{R} \rangle\subseteq \langle G_{In(\rho)} \rangle$, we have proved that $$\langle \widehat{B}^{-1}T_{R-rk} \rangle\subseteq \langle G_{In(\rho)} \rangle.$$
Hence, by Theorem \ref{thm:condition}, we obtain that the constructed code $\widehat{\mathcal{C}}=\widehat{B} \cdot \mathcal{C}$ is an admissible $(R-rk, k)$ linear secure network code for the model $(\mathcal{N}, f,  \mathcal{W}_{r},\zeta)$. 

Remark that the constructions in \eqref{eq:choice1} and \eqref{eq:choice2} not only ensure the security of the security function $\zeta$, but also guarantee the security of all source messages, which is unnecessary. Therefore, for the condition in Theorem \ref{thm:condition}, the constructions in \eqref{eq:choice1} and \eqref{eq:choice2} are only sufficient, not necessary. This point will be illustrated by an example in the last subsection.

\subsection{The Required Field Size of the Code Construction}

In this subsection, we establish upper bounds on the minimum required field size for the existence of an admissible $(R-rk, k)$ linear secure network code in the secure model $(\mathcal{N}, f, \mathcal{W}_{r},\zeta)$, where $f(\vec{m}_{S})=\vec{m}_{S}\cdot T$, $\zeta(\vec{m}_{S})=\vec{m}_{S}\cdot \Upsilon$ and the security level $r$ satisfies $1\leq rk\leq R$.

We define a set:
    \[\widehat{\mathcal{W}}_{r}'=\{W\in \widehat{\mathcal{W}}_{r}^{*} :|W|=r\}.\]
The following lemma demonstrates that to satisfy the security condition \eqref{eq:security condition}, it is sufficient to consider the significantly reduced collection of wiretap sets $\widehat{\mathcal{W}}_{r}'$ instead of the full set $\mathcal{W}_{r}$. This result can be obtained by a proof similar to that in \cite[Theorem 9]{arxiv2025secure}, and thus the proof is omitted here.
\begin{lemma}
Consider the secure network function computation model $(\mathcal{N}, f, \mathcal{W}_{r}, \zeta)$. A secure network code $\mathcal{C}$ for $(\mathcal{N}, f, \mathcal{W}_{r}, \zeta)$ satisfies the security condition \eqref{eq:security condition} if and only if
    \begin{align}\label{eq:security condition1}
    I(Y_{W}; \zeta(\mathbf{M}_{S}))=0, \ \ \forall\  W \in \widehat{\mathcal{W}}_{r}'.
    \end{align}
\end{lemma}

By Lemma \ref{lem:security wen}, one of the conditions in Theorem \ref{thm:condition} can be replaced by
\[
 \langle \widehat{B}^{-1} \cdot \mathbf{S} \rangle \cap \langle G_W \rangle = \{ \vec{0} \} \quad \forall\ W \in \widehat{\mathcal{W}}_{r}'.
\]
Consequently, the selection condition for $\vec{b}_{1}, \vec{b}_{2},\cdots, \vec{b}_{R-rk}$ can be revised as follows for $1 \leq j \leq R-rk$:
\begin{align}\label{eq:choose3}\vec{b}_{j} \in \mathbb{F}_{q}^{R} \backslash \bigcup_{W \in \widehat{\mathcal{W}}_{r}'} \bigcup_{i=1}^{s}\left(\left\langle G_{W}^{(\sigma_{i})}\right\rangle+\left\langle \vec{b}_{1}, \vec{b}_{2}, \cdots, \vec{b}_{j-1}\right\rangle\right).\end{align}

\begin{theorem}\label{thm:q}
Consider the secure network function computation model $(\mathcal{N}, f, \mathcal{W}_{r}, \zeta)$, where $f(\vec{m}_{S})=\vec{m}_{S}\cdot T$, $\zeta(\vec{m}_{S})=\vec{m}_{S}\cdot \Upsilon$, and the security level $r$ satisfies $0 \leq rk \leq C_{\text{min}}$. Then, for any nonnegative integer $R$ with $rk \leq R \leq  C_{\text{min}}$, there exists an $\mathbb{F}_{q}$-valued admissible $(R-rk, k)$ linear secure network code for $(\mathcal{N}, f, \mathcal{W}_{r}, \zeta)$ if the field size $q$ satisfies
    \[
    q> s\cdot |\widehat{\mathcal{W}}_{r}'|.
    \]
\end{theorem}

\begin{proof}
For notational simplicity, define
\[\mathcal{B}_{j-1}=\left\langle \vec{b}_{1}, \vec{b}_{2}, \cdots, \vec{b}_{j-1}\right\rangle, \forall 1 \leq j \leq R.\]
For $1 \leq j \leq R-rk$, we have
    \begin{align}
    & \left|\mathbb{F}_{q}^{R} \backslash \bigcup_{W \in \widehat{\mathcal{W}}_{r}'} \bigcup_{i=1}^{s}\left(\left\langle G_{W}^{(\sigma_{i})}\right\rangle+\mathcal{B}_{j-1}\right)\right| \notag\\
    & =\left|\mathbb{F}_{q}^{R}\right|-\left|\bigcup_{W \in \widehat{\mathcal{W}}_{r}'} \bigcup_{i=1}^{s}\left(\left\langle G_{W}^{(\sigma_{i})}\right\rangle+\mathcal{B}_{j-1}\right)\right| \notag\\
    & \geq q^{R}-\sum_{W \in \widehat{\mathcal{W}}_{r}'} \sum_{i=1}^{s}\left|\left\langle G_{W}^{(\sigma_{i})}\right\rangle+\mathcal{B}_{j-1}\right|\notag\\
    & \geq q^{R}-\sum_{W \in \widehat{\mathcal{W}}_{r}'} \sum_{i=1}^{s} q^{R-1} \label{eq:at most}\\
    &=q^{R-1} \left(q-s \cdot\left|\widehat{\mathcal{W}}_{r}'\right|\right) >0,\notag
    \end{align}
where the inequality \eqref{eq:at most} holds because for any $1 \leq i \leq s$, $1 \leq j \leq R-rk$, and $W \in \widehat{\mathcal{W}}_{r}'$,
    \[\dim\left(\left\langle G_{W}^{(\sigma_{i})}\right\rangle\right) \leq rk \quad \text{and} \quad \dim\left(\mathcal{B}_{j-1}\right) \leq R-rk-1.\]
For $R-rk+1 \leq j \leq R$, we have
    \[|\mathbb{F}_{q}^{R} \backslash \mathcal{B}_{j-1} |=q^{R}-q^{j-1} \geq q^{R}-q^{R-1}>0 .\]
This completes the proof of the theorem.
\end{proof}

Remark that if the size of the field $\mathbb{F}_{q}$, on which both the linear target function and the security function are defined, is too small to support the proposed code design, we can adopt an extension field $\mathbb{F}_{q^{L}}$ of $\mathbb{F}_{q}$ such that the size $q^{L}$ is sufficient for constructing an $\mathbb{F}_{q^{L}}$-valued admissible $(R-rk, k)$ linear secure network code by our construction method. The extension field $\mathbb{F}_{q^{L}}$ can be regarded as an $L$-dimensional vector space over $\mathbb{F}_{q}$ with a basis $\{1, \alpha, \alpha^{2}, \cdots, \alpha^{L-1}\}$, where $\alpha$ denotes a primitive element of $\mathbb{F}_{q^{L}}$. Hence, the $\mathbb{F}_{q^{L}}$-valued $(R-rk, k)$ linear secure network code for computing $f(\vec{m}_{S})=\vec{m}_{S}\cdot T$ on $\mathbb{F}_{q^{L}}$ can be equivalently interpreted as an $((R-rk) L, kL)$ linear secure network code for computing the same function on $\mathbb{F}_{q}$, while achieving the same secure computing rate $\frac{R}{k}-r$. As a consequence, for computing $f(\vec{m}_{S})=\vec{m}_{S}\cdot T$ defined over an arbitrary finite field, the proposed code construction always yields an $\mathbb{F}_{q}$-valued admissible linear secure network code of rate up to $\frac{C_{\text{min }}}{k}-r$ for the model $(\mathcal{N}, f, \mathcal{W}_{r},\zeta)$ with the security level $r$ satisfying $0 \leq rk \leq C_{\text{min}}$. This implies a lower bound on the secure computing capacity $\hat{\mathcal{C}}(\mathcal{N}, f, \mathcal{W}_{r},\zeta)$, which is formally stated below. 

\begin{theorem}
For the secure network function computation model $(\mathcal{N}, f, \mathcal{W}_{r}, \zeta)$, where $f(\vec{m}_{S})=\vec{m}_{S}\cdot T$, $\zeta(\vec{m}_{S})=\vec{m}_{S}\cdot \Upsilon$, and the security level $r$ satisfies $0 \leq rk \leq C_{\text{min}}$, the secure computing capacity $\widehat{\mathcal{C}}(\mathcal{N}, f, \mathcal{W}_{r},\zeta)$ satisfies 
 \[\widehat{\mathcal{C}}(\mathcal{N}, f, \mathcal{W}_{r},\zeta) \geq \frac{C_{\min }}{k}-r.\]
\end{theorem}

\subsection{An Example}
In this subsection, we present an example to demonstrate our code construction. This example illustrates that the codes derived from the construction in Theorem \ref{thm:condition},
which achieve security-function security, are generally not secure with respect to the source messages. Furthermore, it shows that the bound $s\cdot |\widehat{\mathcal{W}}_{r}'|$ on the field size given in Theorem \ref{thm:q} is merely a sufficient condition and is generally not necessary.

\begin{figure}[htbp]
    \centering
    \begin{tikzpicture}[
        main node/.style={circle, draw, minimum size=12pt, inner sep=1pt},
        every edge/.style={draw, -Stealth, thick}
    ]
    \node[main node] (s1) at (-4.5, 4) {$\sigma_1$};
    \node[main node] (s2) at (0, 4)    {$\sigma_2$};
    \node[main node] (s3) at (4.5, 4)  {$\sigma_3$};
    
    \node[main node] (m1) at (-2, 2.5) {};
    \node[main node] (m2) at (0, 2.5)  {};
    \node[main node] (m3) at (2, 2.5)  {};

    \node[main node] (p1) at (-2, 0.5) {};
    \node[main node] (p2) at (0, 0.5)  {};
    \node[main node] (p3) at (2, 0.5)  {};

    \node[main node] (n1) at (-3, -1.5) {};
    \node[main node] (n2) at (0, -1.5)  {};
    \node[main node] (n3) at (3, -1.5)  {};
    
    \node[main node] (sink) at (0, -3) {$\rho$};

    \node[left=0.2cm of s1] {$x_1, x_2, x_3$};
    \node[above=0.2cm of s2] {$y_1, y_2, y_3$};
    \node[right=0.2cm of s3] {$z_1, z_2, z_3$};

    \path (s1) edge node[above left=2pt] {$e_1$} node[above left=12pt] {$2x_1+2x_2+x_3$} (n1);
    \path (s1) edge node[above left=2pt] {$e_2$} node[below right=2pt] {$x_2$} (m1);
    \path (s1) edge node[above=10pt] {$e_3$} node[above=0.1pt] {$x_1$} (m2);
    
    \path (s2) edge node[above right =2pt] {$e_4$} node[right =2pt] {$y_3$} (m1);
    \path (s2) edge[bend left=15] node[above right=4pt] {$e_5$} node[right=-3pt] {$2y_1+y_2+2y_3$} (n2);
    \path (s2) edge node[above left=1pt] {$e_6$} node[ left=1pt] {$y_1$} (m3);
    
    \path (s3) edge node[above right=2pt] {$e_7$} node[above =0.1pt] {$2z_3$} (m2);
    \path (s3) edge node[above right=1pt] {$e_8$} node[below =2pt] {$2z_2$} (m3);
    \path (s3) edge node[right=2pt] {$e_9$} node[above right=12pt] {$2z_1+z_2+z_3$} (n3);
    
    \path (m1) edge node[left=2pt] {$e_{10}$} node[above left=2pt] {$x_2+y_3$} (p1);
    \path (m2) edge node[left=2pt] {$e_{11}$} node[above left=2pt] {$x_1+2z_3$} (p2);
    \path (m3) edge node[above left=2pt] {$e_{12}$} node[above right=2pt] {$y_1+2z_2$} (p3);
    
    \path (p1) edge node[above left=2pt] {$e_{13}$} node[left=2pt] {$$} (n1);
    \path (p1) edge node[above left=4pt] {$e_{14}$} node[right=-5pt] {$$} (n2);
    \path (p2) edge node[above right=5pt] {$e_{15}$} node[below left=-3pt] {$$} (n1);
    
    \path (p2) edge node[above left=5pt] {$e_{16}$} node[below right=2pt] {$$} (n3);
    \path (p3) edge node[above right=5pt] {$e_{17}$} node[below left=2pt] {$$} (n2);
    \path (p3) edge node[above right=1pt] {$e_{18}$} node[left=2pt] {$$} (n3);
    
    \path (n1) edge node[left=2pt] {$e_{19}$} node[below left =5pt] {$x_3+y_3+2z_3$} (sink);
    \path (n2) edge node[left=2pt] {$e_{20}$} node[above right=1pt] {$x_2+y_2+2z_2$} (sink);
    \path (n3) edge node[right=2pt] {$e_{21}$} node[below right=5pt] {$x_1+y_1+2z_1$} (sink);

    \end{tikzpicture}
    \caption{An $\mathbb{F}_3$-valued $(3, 1)$ linear network code for the model $(\mathcal{N}, f)$, where $y_{e_{10}}=y_{e_{13}}=y_{e_{14}}$, $y_{e_{11}}=y_{e_{15}}=y_{e_{16}}$ and $y_{e_{12}}=y_{e_{17}}=y_{e_{18}}$.}
    \label{fig:base_code_f3}
\end{figure}
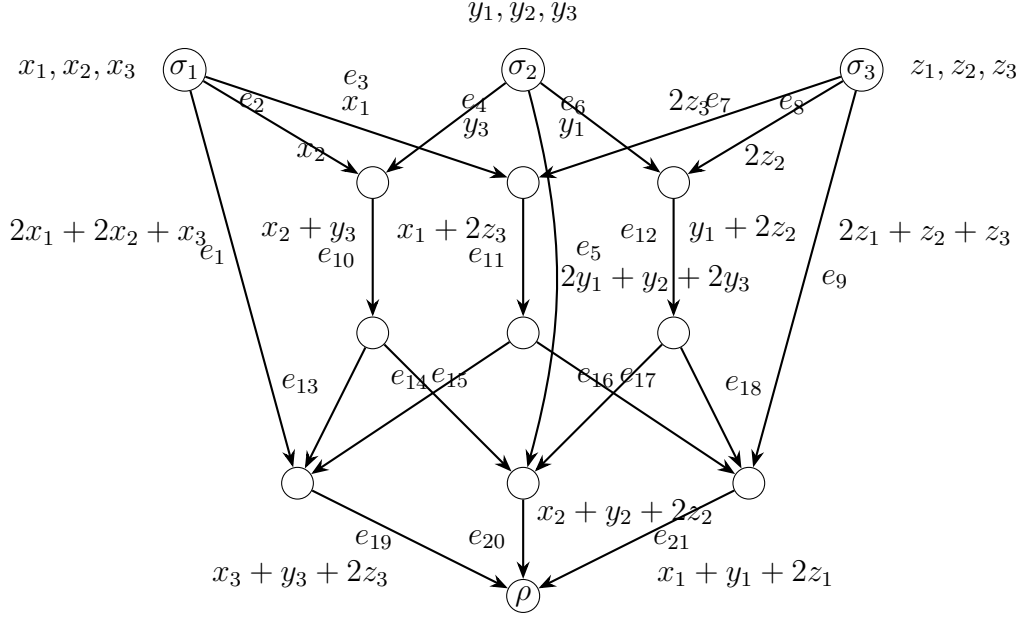

\begin{example}
We consider a secure network function computation model \((\mathcal{N}, f, \mathcal{W}_r, \zeta)\),
where \(\mathcal{N} = (\mathcal{G}, S, \rho)\) is a variant of the reverse butterfly network, as depicted in \autoref{fig:base_code_f3}, with the set of source nodes \(S = \{\sigma_1, \sigma_2, \sigma_3\}\). The target function \(f\) is a scalar linear function over \(\mathbb{F}_3\), specifically,
\[
f(\vec{m}_S) = m_1 + m_2 + 2m_3.
\]
The security function \(\zeta\) is a vector-linear function over \(\mathbb{F}_3\), given by \(\zeta(\vec{m}_S) = \vec{m}_S \cdot \Upsilon\), where
\[
\Upsilon = \begin{pmatrix}
1 & 1 \\
1 & 0 \\
2 & 0
\end{pmatrix}.
\]
Furthermore, the security level is \(r = 1 < C_{\min} = 3\). Below, we will construct an admissible $(2,1)$ linear secure network code \(\hat{\mathcal{C}}\) for the model \((\mathcal{N}, f, \mathcal{W}_r, \zeta)\) using our proposed construction. By Theorem \ref{thm:the upper bound}, we can obtain \(\hat{\mathcal{C}}(\mathcal{N}, f, \mathcal{W}_r, \zeta) \leq 2\). This implies that \(\hat{\mathcal{C}}(\mathcal{N}, f, \mathcal{W}_r, \zeta) = 2\). Thus, the constructed code \(\hat{\mathcal{C}}\) is optimal.

First, we have an \(\mathbb{F}_3\)-valued \((3, 1)\) linear network code \(\mathcal{C}\) for the model \((\mathcal{N}, f)\), as shown in \autoref{fig:base_code_f3}. Evidently, this code fails to ensure the secrecy of the security function. The global encoding vectors of the code \(\mathcal{C}\) are:
\[
\begin{aligned}
g_{e_1} &= \left(\begin{smallmatrix} 2 \\ 2 \\ 1 \\ 0 \\ 0 \\ 0 \\ 0 \\ 0 \\ 0 \end{smallmatrix}\right), 
g_{e_2} = \left(\begin{smallmatrix} 0 \\ 1 \\ 0 \\ 0 \\ 0 \\ 0 \\ 0 \\ 0 \\ 0 \end{smallmatrix}\right), 
g_{e_3} = \left(\begin{smallmatrix} 1 \\ 0 \\ 0 \\ 0 \\ 0 \\ 0 \\ 0 \\ 0 \\ 0 \end{smallmatrix}\right), 
g_{e_4} = \left(\begin{smallmatrix} 0 \\ 0 \\ 0 \\ 0 \\ 0 \\ 1 \\ 0 \\ 0 \\ 0 \end{smallmatrix}\right), 
g_{e_5} = \left(\begin{smallmatrix} 0 \\ 0 \\ 0 \\ 2 \\ 1 \\ 2 \\ 0 \\ 0 \\ 0 \end{smallmatrix}\right), 
g_{e_6} = \left(\begin{smallmatrix} 0 \\ 0 \\ 0 \\ 1 \\ 0 \\ 0 \\ 0 \\ 0 \\ 0 \end{smallmatrix}\right), 
g_{e_7} = \left(\begin{smallmatrix} 0 \\ 0 \\ 0 \\ 0 \\ 0 \\ 0 \\ 0 \\ 0 \\ 2 \end{smallmatrix}\right), \\
g_{e_8} & = \left(\begin{smallmatrix} 0 \\ 0 \\ 0 \\ 0 \\ 0 \\ 0 \\ 0 \\ 2 \\ 0 \end{smallmatrix}\right), 
g_{e_9} = \left(\begin{smallmatrix} 0 \\ 0 \\ 0 \\ 0 \\ 0 \\ 0 \\ 2 \\ 1 \\ 1 \end{smallmatrix}\right), 
g_{e_{10}} = g_{e_{13}} = g_{e_{14}} = \left(\begin{smallmatrix} 0 \\ 1 \\ 0 \\ 0 \\ 0 \\ 1 \\ 0 \\ 0 \\ 0 \end{smallmatrix}\right), 
g_{e_{11}} = g_{e_{15}} = g_{e_{16}} = \left(\begin{smallmatrix} 1 \\ 0 \\ 0 \\ 0 \\ 0 \\ 0 \\ 0 \\ 0 \\ 2 \end{smallmatrix}\right), \\
g_{e_{12}} &= g_{e_{17}} = g_{e_{18}} = \left(\begin{smallmatrix} 0 \\ 0 \\ 0 \\ 1 \\ 0 \\ 0 \\ 0 \\ 2 \\ 0 \end{smallmatrix}\right), 
g_{e_{19}} = \left(\begin{smallmatrix} 0 \\ 0 \\ 1 \\ 0 \\ 0 \\ 1 \\ 0 \\ 0 \\ 2 \end{smallmatrix}\right), 
g_{e_{20}} = \left(\begin{smallmatrix} 0 \\ 1 \\ 0 \\ 0 \\ 1 \\ 0 \\ 0 \\ 2 \\ 0 \end{smallmatrix}\right), 
g_{e_{21}} = \left(\begin{smallmatrix} 1 \\ 0 \\ 0 \\ 1 \\ 0 \\ 0 \\ 2 \\ 0 \\ 0 \end{smallmatrix}\right).
\end{aligned}
\]
We can see that the code \(\mathcal{C}\) can compute the target function \(f\) three times,
but it does not guarantee the security of the security function. For example, \(x_1 + y_1 + 2z_1\) can be obtained when the information on edge \(e_{21}\) is eavesdropped.

Based on the code \(\mathcal{C}\), we construct an \(\mathbb{F}_3\)-valued \((2, 1)\) linear secure network code \(\hat{\mathcal{C}}\) for the model \((\mathcal{N}, f, \mathcal{W}_r, \zeta)\)
by selecting an appropriate transformation matrix \(\widehat{B}\).
First, we consider the construction scheme in \eqref{eq:choice1} and \eqref{eq:choice2}.
We need to select three vectors \(\vec{b}_1, \vec{b}_2, \vec{b}_3\) according to the following rules
to form \(\widehat{B} =\left[ B \quad B \quad  B \right]_{3}^{\text{diagonal}}\), where \(B^{-1} = \left[ \vec{b}_1 \quad \vec{b}_2 \quad \vec{b}_3\right]\) is invertible:
\begin{enumerate}
    \item\label{cond:b1} \(\vec{b}_1 \in \mathbb{F}_3^3 \setminus \bigcup_{j=1}^{21} \bigcup_{i=1}^{3}\langle g_{e_j}^{\sigma_{i}} \rangle\),
    \item \(\vec{b}_2 \in \mathbb{F}_3^3 \setminus \bigcup_{j=1}^{21}\bigcup_{i=1}^{3} \left( \langle g_{e_j}^{\sigma_{i}} \rangle + \langle \vec{b}_1 \rangle \right)\),
    \item \(\vec{b}_3 \in \mathbb{F}_3^3 \setminus \langle \vec{b}_1, \vec{b}_2 \rangle\).
\end{enumerate}
Therefore, we can choose \(\vec{b}_1 =\left( \begin{smallmatrix} 1 \\ 1\\0 \end{smallmatrix}\right)\), which satisfies Condition \ref{cond:b1}.
When selecting such \(\vec{b}_1\), we find that
\[
\mathbb{F}_3^3 \setminus \bigcup_{j=1}^{21}\bigcup_{i=1}^{3} \left( \langle g_{e_j}^{\sigma_{i}}\rangle + \langle \left( \begin{smallmatrix} 1 \\ 1\\0 \end{smallmatrix}\right) \rangle \right) = \emptyset,
\]
thus the construction scheme in \eqref{eq:choice1} and \eqref{eq:choice2} is ineffective in this case.

Next, we consider the construction scheme in Theorem \ref{thm:condition}.
We need to select three vectors \(\vec{b}_1, \vec{b}_2, \vec{b}_3\) to form \(\widehat{B} =\left[ B \quad B \quad  B \right]_{3}^{\text{diagonal}}\),
with \(B^{-1} = \left[ \vec{b}_1 \quad \vec{b}_2 \quad \vec{b}_3\right]\) invertible, which must satisfy
\[
\langle \widehat{B}^{-1}\cdot \mathbf{S} \rangle \cap \langle g_{e_j} \rangle = \{\vec{0}\}, \quad \forall 1 \leq j \leq 21.
\]
This is equivalent to
\[
\langle \Upsilon \otimes \left[\vec{b}_1\quad \vec{b}_2\right] \rangle \cap \langle g_{e_j} \rangle = \{\vec{0}\}, \quad \forall 1 \leq j \leq 21.
\]
For instance, we can choose
\[
\vec{b}_1 = \begin{pmatrix} 1 \\ 1\\0 \end{pmatrix}, \quad \vec{b}_2 = \begin{pmatrix} 0\\1 \\ 1 \end{pmatrix}, \quad \vec{b}_3 = \begin{pmatrix} 0\\1 \\ 0 \end{pmatrix}.
\]
Then we have
\[
B^{-1} = \left[ \vec{b}_1 \quad \vec{b}_2 \quad \vec{b}_3\right] = \begin{pmatrix} 1 & 0 & 0 \\1 & 1 & 1 \\0&1&0\end{pmatrix},
\]
and thus
\[
\widehat{B} =\left[ B \quad B \quad  B \right]_{3}^{\text{diagonal}} \quad \text{with} \quad B = \begin{pmatrix} 1 & 0 & 0 \\0 & 0 & 1 \\2&1&2\end{pmatrix}.
\]

Based on the above construction, we obtain an \(\mathbb{F}_3\)-valued \((2, 1)\) linear secure network code \(\hat{\mathcal{C}}=\widehat{B}\cdot \mathcal{C}\) for \((\mathcal{N}, f, \mathcal{W}_r, \zeta)\), where the global encoding vector is \(\hat{g}_{e_j} = \widehat{B}\cdot g_{e_j}\) for \(1 \leq j \leq 21\), i.e.,
\[
\begin{aligned}
\hat{g}_{e_1} &= \left(\begin{smallmatrix} 2 \\ 1 \\ 2 \\ 0 \\ 0 \\ 0 \\ 0 \\ 0 \\ 0 \end{smallmatrix}\right), 
\hat{g}_{e_2} = \left(\begin{smallmatrix} 0 \\ 0 \\ 1 \\ 0 \\ 0 \\ 0 \\ 0 \\ 0 \\ 0 \end{smallmatrix}\right), 
\hat{g}_{e_3} = \left(\begin{smallmatrix} 1 \\ 0 \\ 2 \\ 0 \\ 0 \\ 0 \\ 0 \\ 0 \\ 0 \end{smallmatrix}\right), 
\hat{g}_{e_4} = \left(\begin{smallmatrix} 0 \\ 0 \\ 0 \\ 0 \\ 1 \\ 2 \\ 0 \\ 0 \\ 0 \end{smallmatrix}\right), 
\hat{g}_{e_5} = \left(\begin{smallmatrix} 0 \\ 0 \\ 0 \\ 2 \\ 2 \\ 0 \\ 0 \\ 0 \\ 0 \end{smallmatrix}\right), 
\hat{g}_{e_6} = \left(\begin{smallmatrix} 0 \\ 0 \\ 0 \\ 1 \\ 0 \\ 2 \\ 0 \\ 0 \\ 0 \end{smallmatrix}\right), 
\hat{g}_{e_7} = \left(\begin{smallmatrix} 0 \\ 0 \\ 0 \\ 0 \\ 0 \\ 0 \\ 0 \\ 2 \\ 1 \end{smallmatrix}\right), \\
\hat{g}_{e_8} & = \left(\begin{smallmatrix} 0 \\ 0 \\ 0 \\ 0 \\ 0 \\ 0 \\ 0 \\ 0 \\ 2 \end{smallmatrix}\right), 
\hat{g}_{e_9} = \left(\begin{smallmatrix} 0 \\ 0 \\ 0 \\ 0 \\ 0 \\ 0 \\ 2 \\ 1 \\ 1 \end{smallmatrix}\right), 
\hat{g}_{e_{10}} = g_{e_{13}} = g_{e_{14}} = \left(\begin{smallmatrix} 0 \\ 0 \\ 1 \\ 0 \\ 1 \\ 2 \\ 0 \\ 0 \\ 0 \end{smallmatrix}\right), 
\hat{g}_{e_{11}} = g_{e_{15}} = g_{e_{16}} = \left(\begin{smallmatrix} 1 \\ 0 \\ 2 \\ 0 \\ 0 \\ 0 \\ 0 \\ 2 \\ 1 \end{smallmatrix}\right), \\
\hat{g}_{e_{12}} &= g_{e_{17}} = g_{e_{18}} = \left(\begin{smallmatrix} 0 \\ 0 \\ 0 \\ 1 \\ 0 \\ 2 \\ 0 \\ 0 \\ 2 \end{smallmatrix}\right), 
\hat{g}_{e_{19}} = \left(\begin{smallmatrix} 0 \\ 1 \\ 2 \\ 0 \\ 1 \\ 2 \\ 0 \\ 2 \\ 1 \end{smallmatrix}\right), 
\hat{g}_{e_{20}} = \left(\begin{smallmatrix} 0 \\ 0 \\ 1 \\ 0 \\ 0 \\ 1 \\ 0 \\ 0 \\ 2 \end{smallmatrix}\right), 
\hat{g}_{e_{21}} = \left(\begin{smallmatrix} 1 \\ 0 \\ 2 \\ 1 \\ 0 \\ 2 \\ 2 \\ 0 \\ 1 \end{smallmatrix}\right).
\end{aligned}
\]
The encoding process of the code \(\hat{\mathcal{C}}\) is illustrated in \autoref{fig:base_code_f4},
where \(m_{ij}\) denotes the source information generated by the source node \(\sigma_i\) and \(k_i\) denotes the random key generated by \(\sigma_i\).
We can verify that the code \(\hat{\mathcal{C}}\) satisfies both the computability and security conditions.
Specifically, by adding the information on edges \(y_{e_{20}}, y_{e_{21}}\), we can obtain \(m_{11} + m_{21} + 2m_{31}\), and by adding \(y_{e_{19}}\) and \(y_{e_{20}}\), we can obtain \(m_{12} + m_{22} + 2m_{32}\). Furthermore, it is easy to see that the eavesdropper cannot obtain any information about the target function or the source information generated by $\sigma_{1}$
by eavesdropping on any single edge in the network.
\end{example}

\begin{figure}[htbp]
    \centering
    \begin{tikzpicture}[
        main node/.style={circle, draw, minimum size=12pt, inner sep=1pt},
        every edge/.style={draw, -Stealth, thick}
    ]
    \node[main node] (s1) at (-4.5, 4) {$\sigma_1$};
    \node[main node] (s2) at (0, 4)    {$\sigma_2$};
    \node[main node] (s3) at (4.5, 4)  {$\sigma_3$};
    
    \node[main node] (m1) at (-2, 2.5) {};
    \node[main node] (m2) at (0, 2.5)  {};
    \node[main node] (m3) at (2, 2.5)  {};

    \node[main node] (p1) at (-2, 0.5) {};
    \node[main node] (p2) at (0, 0.5)  {};
    \node[main node] (p3) at (2, 0.5)  {};

    \node[main node] (n1) at (-3, -1.5) {};
    \node[main node] (n2) at (0, -1.5)  {};
    \node[main node] (n3) at (3, -1.5)  {};
    
    \node[main node] (sink) at (0, -3) {$\rho$};

    \node[left=0.2cm of s1] {$(m_{11}, m_{12}, k_{1})$};
    \node[above=0.2cm of s2] {$(m_{21}, m_{22}, k_{2})$};
    \node[right=0.2cm of s3] {$(m_{31}, m_{32}, k_{3})$};

    \path (s1) edge node[above left=2pt] {$e_1$} node[below left=2pt,font=\footnotesize] {$2m_{11}+m_{12}+2k_{1}$} (n1);
    \path (s1) edge node[above left=2pt] {$e_2$} node[below right=2pt,font=\footnotesize] {$k_{1}$} (m1);
    \path (s1) edge node[above=1pt] {$e_3$} node[above =10pt,font=\footnotesize] {$m_{11}+2k_{1}$} (m2);
    
    \path (s2) edge node[above right =2pt] {$e_4$} node[right =-14pt,font=\footnotesize] {$m_{22}+2k_{2}$} (m1);
    \path (s2) edge[bend left=15] node[above right=4pt] {$e_5$} node[right=-3pt,font=\footnotesize] {$2m_{21}+2m_{22}$} (n2);
    \path (s2) edge node[above left=1pt] {$e_6$} node[ above=8pt,font=\footnotesize] {$m_{21}+2k_{2}$} (m3);
    
    \path (s3) edge node[above right=2pt] {$e_7$} node[below =-8pt,font=\footnotesize] {$2m_{32}+k_3$} (m2);
    \path (s3) edge node[above right=1pt] {$e_8$} node[below =2pt,font=\footnotesize] {$2k_3$} (m3);
    \path (s3) edge node[right=2pt] {$e_9$} node[below right=2pt,font=\footnotesize] {$2m_{31}+m_{32}+k_3$} (n3);
    
    \path (m1) edge node[left=2pt] {$e_{10}$} node[above left=2pt, font=\footnotesize] {$ k_1+m_{22}+2k_2$} (p1);
    \path (m2) edge node[left=2pt] {$e_{11}$} node[above =2pt, font=\footnotesize] {$m_{11}+2k_{1}+2m_{32}+k_{3}$} (p2);
    \path (m3) edge node[ right=2pt] {$e_{12}$} node[above right=2pt, font=\footnotesize] {$ m_{21}+2k_{2}+2k_{3}$} (p3);
    
    \path (p1) edge node[above left=2pt] {$e_{13}$} node[left=2pt] {$$} (n1);
    \path (p1) edge node[above left=4pt] {$e_{14}$} node[right=-5pt] {$$} (n2);
    \path (p2) edge node[above right=5pt] {$e_{15}$} node[below left=-3pt] {$$} (n1);
    
    \path (p2) edge node[above left=5pt] {$e_{16}$} node[below right=2pt] {$$} (n3);
    \path (p3) edge node[above right=5pt] {$e_{17}$} node[below left=2pt] {$$} (n2);
    \path (p3) edge node[above right=1pt] {$e_{18}$} node[left=2pt] {$$} (n3);
    
    \path (n1) edge node[left=2pt] {$e_{19}$} node[below left =5pt,font=\footnotesize] {$m_{12}+2k_{1}+m_{22}+2k_{2}+2m_{32}+k_{3}$} (sink);
    \path (n2) edge node[left=2pt] {$e_{20}$} node[above right=1pt,font=\footnotesize] {$k_{1}+k_{2}+2k_{3}$} (sink);
    \path (n3) edge node[right=2pt] {$e_{21}$} node[below right=5pt,font=\footnotesize] {$m_{11}+2k_{1}+m_{21}+2k_{2}+2m_{31}+k_{3}$} (sink);

    \end{tikzpicture}
    \caption{An $\mathbb{F}_3$-valued $(2, 1)$ linear secure network code for the secure model $(\mathcal{N}, f, \mathcal{W}_{r},\zeta)$, where $y_{e_{10}}=y_{e_{13}}=y_{e_{14}}$, $y_{e_{11}}=y_{e_{15}}=y_{e_{16}}$ and $y_{e_{12}}=y_{e_{17}}=y_{e_{18}}$.}
    \label{fig:base_code_f4}
\end{figure}
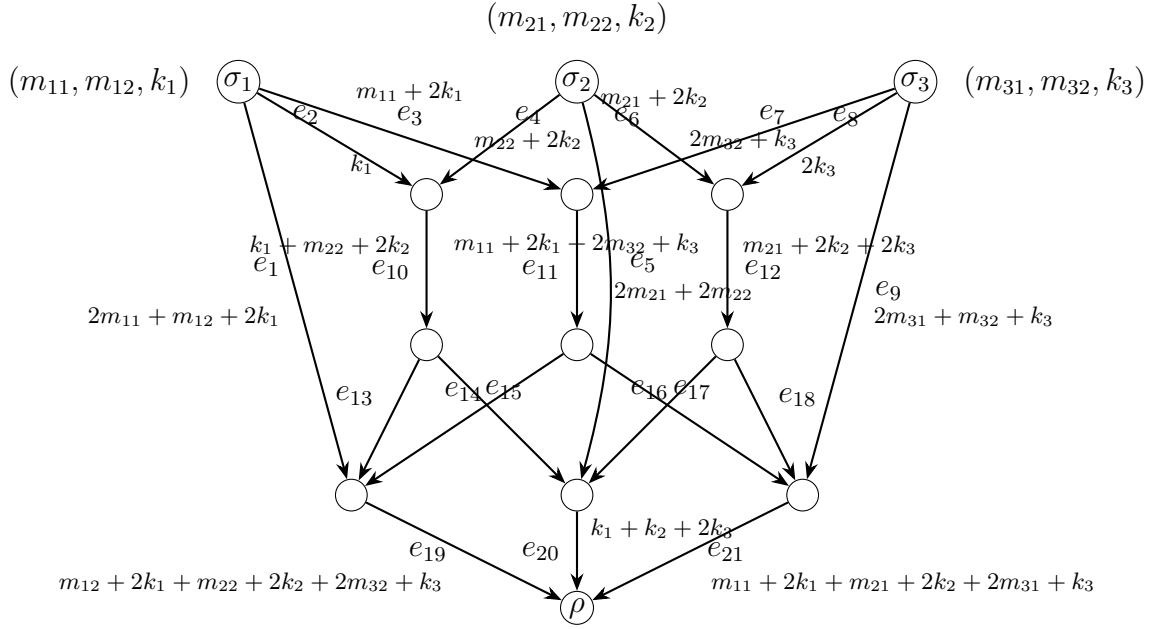

Remark that the code $\hat{\mathcal{C}}$ is not source secure.
For instance, the symbol transmitted on edge $e_5$ is given by $y_{e_5} = 2m_{21} + 2m_{22}$.
In addition, according to Theorem \ref{thm:q},
an $\mathbb{F}_q$-valued admissible $(2,1)$-linear secure network code for $(\mathcal{N}, f, \mathcal{W}_r, \zeta)$
can be constructed whenever $q > |s \cdot \widehat{W}'_r| = 45$.
Nevertheless, our explicit construction shows that $\mathbb{F}_3$ already suffices, i.e., $q=3$.
This illustrates that the bound on the field size presented in Theorem \ref{thm:q}
is only a sufficient condition and is far from necessary.

\section{Conclusion}\label{sec:conclusion}

In this paper, we investigate a general model for secure network function computation, in which both the target function $f$ and the security function $\zeta$ are arbitrary, and an eavesdropper can access any edge subset of size up to a given security level without obtaining any information about the security function.
Using information-theoretic tools, we establish a general upper bound on the secure computing capacity that is valid for arbitrary networks, arbitrary target and security functions, and arbitrary security levels.
In particular, when the target function is the algebraic sum and the security function corresponds to protecting all source messages, our general bound reduces to that in \cite{GBY2024}; when both functions are algebraic sums, it recovers the upper bound given in \cite{arxiv2025secure}.

Furthermore, we study the secure network function computation model where the target function is vector-linear and the security function requires secrecy of all source messages.
Since the corresponding upper bound on the secure computing capacity is graph-theoretic and inherently complicated, we derive a simplified form using order-theoretic methods.
Based on existing algorithms, we further propose an efficient algorithm to compute this simplified upper bound, with time complexity linear in the number of edges of the network.

Finally, we consider the setting where both the target function and the security function are vector-linear.
By characterizing the equivalent conditions for the computability and security of linear secure network codes, we develop two explicit code constructions and derive an upper bound on the minimum finite field size required for our constructions.
Accordingly, we obtain a nontrivial lower bound on the secure computing capacity for this general vector-linear setting.
We note that under the same setting, a tighter upper bound on the secure computing capacity has been established in the very recent work \cite{xu2026network}.

Several interesting and promising problems remain for future research.
First, we intend to design improved linear secure network code constructions that achieve smaller required field sizes or higher secure computing rates, which is of practical importance in secure network function computation.
Second, we plan to investigate the fundamental limits on random key resources for secure computation, including the minimum key size and minimum entropy required at each source node, so as to provide more practical design guidelines.
In the general framework of secure network function computation, we also aim to explore alternative security criteria to accommodate diverse application scenarios.

\ifCLASSOPTIONcaptionsoff
  \newpage
\fi

\bibliographystyle{IEEEtran}  
\bibliography{references}

\end{document}